\documentclass{article}
\usepackage{fullpage}
\usepackage{epsfig}
\usepackage{graphics}
\usepackage{latexsym}
\usepackage{amsmath}
\usepackage{amsfonts}
\usepackage{amssymb}
\usepackage[mathscr]{eucal}
\usepackage{mathrsfs}
\usepackage{pifont}
\usepackage{yhmath}
\usepackage{undertilde}
\usepackage[usenames]{color}
\usepackage{amsthm}
\usepackage[ruled]{algorithm2e}

\usepackage{graphicx}
\usepackage{url}
\urldef{\mailsa}\path|{zhigangcao,xchen,xdhu,wcj}@amss.ac.cn|

\newtheorem{theorem}{Theorem}[section]

\newtheorem{definition}[theorem]{Definition}

\newtheorem{example}[theorem]{Example}

 \theoremstyle{remark}

\newtheorem{clm}{Claim}{\itshape}{\rmfamily}
\newtheorem{observation}[theorem]{Observation}

{\itshape}{\rmfamily}

\newtheorem{remark}[theorem]{Remark}

\makeatletter
\@addtoreset{equation}{section}
\def\section{\@startsection {section}{1}{\z@}{-3.5ex plus -1ex minus
 -.2ex}{2.3ex plus .2ex}{\large\bf}}
\makeatother

\def\bfm#1{\mbox{\boldmath$#1$}}

\def\0{\bfm 0}

\newcommand{\qedd}{\hfill\rule{1 ex}{1 ex}}

\newcommand{\redcomment}[1]{\textcolor{black}{\textrm{#1}}}

\newcommand{\msc}[1]{\ensuremath{\text{\sc #1}}\xspace}

\newcommand{\rv}{\msc{r}}
\newcommand{\opt}{\msc{opt}}

\DeclareMathAlphabet{\mathpzc}{OT1}{pzc}{m}{it}

\def\bfm#1{\mbox{\boldmath$#1$}}

\begin{document}


\title{\bf Pricing in Social Networks with Negative Externalities \thanks{Supported in part by NNSF of China under Grant No. 11222109, 71101140 and 11471326,  973 Project of China under Grant No. 2011CB80800 and 2010CB731405, and  CAS Program for Cross \&
Cooperative Team of Science \& Technology Innovation.}}


 \author{Zhigang Cao  \and Xujin Chen \and Xiaodong Hu
 \and
 Changjun Wang}
\date{Academy of Mathematics and Systems Science \\ Chinese Academy
of Sciences, Beijing 100190, China\\
${}$\\
\mailsa}

\maketitle

\begin{abstract}
We study the problems of pricing an indivisible product  to   consumers who are embedded in a {given} social network. {The goal is to maximize the revenue {of the seller}. We assume impatient consumers who buy the product as soon as the seller posts a price not greater than their {valuations} of the product.} {The 
{product's value for a consumer} 
{is determined by two factors:} 
a fixed {consumer-specified} intrinsic value  and a variable externality that is exerted  from {the consumer's} 
neighbors in a linear way}. {We study the scenario of negative  externalities, which captures many interesting situations, 
but is much less understood {in comparison} 
with its positive externality counterpart.}
{We assume complete information about the network, consumers' intrinsic values, and the {negative externalities. The maximum revenue is in general achieved by iterative pricing,}}  
{
which offers  impatient consumers a sequence of prices
over time.}


  We prove that it is NP-hard to find an optimal iterative pricing, even for unweighted tree networks {with uniform intrinsic values}.
  {Complementary to the hardness result,} we design a  2-approximation algorithm for finding iterative pricing in
general weighted networks with {(possibly)} nonuniform intrinsic values. We show that, as an approximation to optimal iterative pricing, single pricing {works} rather well for many interesting cases, {such as forests, Erd\H{o}s-R\'enyi networks  and Barab\'asi-Albert networks}, {although its worst-case performance can be} arbitrarily bad.
\end{abstract}


\newcounter{my}
\newenvironment{mylabel}
{
\begin{list}{(\roman{my})}{
\setlength{\parsep}{-0mm}
\setlength{\labelwidth}{8mm}
\setlength{\leftmargin}{8mm}
\usecounter{my}}
}{\end{list}}

\newcounter{my2}
\newenvironment{mylabel2}
{
\begin{list}{(\alph{my2})}{
\setlength{\parsep}{-1mm} \setlength{\labelwidth}{12mm}
\setlength{\leftmargin}{14mm}
\usecounter{my2}}
}{\end{list}}

\newcounter{my3}
\newenvironment{mylabel3}
{
\begin{list}{(\alph{my3})}{
\setlength{\parsep}{-1mm}
\setlength{\labelwidth}{8mm}
\setlength{\leftmargin}{10mm}
\usecounter{my3}}
}{\end{list}}

\section{Introduction}
{People  interact} with and influence each other to a degree that is beyond most of us can imagine. The magnitude of this
connection has been upgraded to a brandnew level by the proliferation of online SNS (Social Network Services, e.g. Facebook, Twitter, Google Plus,
and SinaWeibo). Numerous business opportunities are being incubated by this upgrading. Yet, its consequences are far from being fully unfolded or
understood, leaving many fascinating questions for scientists in a variety of disciplines to answer. One incredible fact in the SNS era is that we
are now able to know the complete network of  who is connected with whom. Network marketing and pricing, {with the assistance of {\em big data},}
could be much more precise and flexible than traditional {counterparts, and} are attracting increasing attention from both industry and academia. In this paper, we
 study, from an algorithmic point of view, {how  a monopolist seller should price to the} consumers connected by a known social network.

Consumption is never a completely private thing. As opposed to standard economic settings, the utilities that a consumer obtains from
consuming many {kinds of} goods, are not determined merely by his/her private needs and the functions and qualities of the goods, {but   also} greatly
affected by the consumptions of {his/her social network neighbors}.
 {For example, the} reason that we wear clothes is not only to cover ourselves from {cold}, but {usually} also to make
other people think that we look great and {unique}. 
 This social side of consumption is becoming more and more prominent with the
unification of E-commerce and SNS. It is now {very convenient for us} to share with our friends our shopping results. By clicking one more button at
the time we pay for the skirt online, all our Twitter friends know immediately the complete information of this skirt. This effect could be much
stronger and faster than face-to-face sharing. Our ladybros may think the skirt terrific and get one too, or {oppositely, they may prefer later a different} style to avoid outfit clash. The former case is typical {\em positive
externality}: the incentive that a consumer buys a product increases as more and more  {of} his/her social network neighbors buy the {product.} The latter opposite
scenario, {the incentive decreases when more neighbors have the product, is referred to as} {\em negative externality}, {which} is the focus of this paper. Positive externalities are prevalent in many aspects of the society and have
been extensively studies under various academical terms (
herd behavior, Matthew effect, 
strategic complements, and 
viral marketing, {to name a few}). Negative externalities, in contrast, although widely {exist,  are much} less
studied.

\paragraph{\bf Pricing with negative externalities.} We concentrate on the negative externality {among} consumers of consuming a single kind of {product}, 
 {which is usually luxury or fashionable one.} 
{An important} reason that {a consumer buys this product} is to showoff in front of {his/her}
friends (also referred to as {\em invidious consumption} in  {literature}). {Naturally, a consumer buys the product if  the price is not higher than his/her (total) value of the product, which is the sum of his/her constant {\em intrinsic value} and varying  {\em external value}. We propose and study the typical network pricing model, where the external value}   is  the  (weighted) number of people to whom {the consumer} can showoff (i.e. {his/her social network neighbors} who do
not possess this {product}). 
We study, to obtain a maximum revenue, {how   a monopolist seller should} price {such a product with negative externality} 
to consumers connected by a {link-weighted social network, where the revenue is the total payment the seller receives, and the nonnegative integer link weights represent the influences between consumers. While, with the help of SNS, the knowledge of social network structures   and real-time externalities is  available, consumers' intrinsic values might be known in {complete information  scenarios, or partially known in incomplete information  scenarios.  This paper addresses the pricing problems for revenue maximization in complete information scenarios}. Our study falls {into} the framework of {\em uniform pricing}, where at any time point the same take-it-or-leave price is offered (posted) to all consumers who have not bought the product.
{The seller adopts a strategy of {\em iterative pricing} -- posting different prices sequentially at discrete time points, to maximize her revenue (we assume that production costs are zero).}
{ We also assume that the consumers} {are {\em myopic} {(a.k.a. {\em impatient})} in the sense that they, when making purchase decisions, do not} take into account  their neighbors' {future} actions (which might change their external values of the product).}

\paragraph{\bf Contributions.} {Comparing with their positive counterparts, negative externalities possess more irregularity and pose more challenges for research on {product} {diffusion}, especially from the perspective of pricing. The intuitive hardness is   {confirmed by} the {following} theoretical intractability.
 \begin{itemize}
 \item By a reduction from the
  3SAT problem we show that finding an optimal iterative pricing is NP-hard even for the extremely simple
   case of unweighted tree network {with uniform} intrinsic values.
   \end{itemize}
  Complementary to {the hardness result,}
\begin{itemize}
\item We design a 2-approximation algorithm for iterative pricing in general weighted networks with
  general intrinsic values.
    An exact $O(n^2)$-time algorithm is designed for unweighted split networks with uniform intrinsic values.
    \end{itemize}
 The  2-approximation algorithm  is remarkable for its simplicity and versatility to handle the most general problem {regardless of network topologies, link weights or intrinsic values}. {We also study single pricing as an approximation of iterative pricing, and {obtain} the following {negative and positive} results}
 \begin{itemize}
 \item  We prove that optimal single pricing can be arbitrarily worse (at a rate of $\ln\ln n$)
  than the optimal iterative pricing; and on the other hand, optimal single pricing provides nice approximations to
   the optimal iterative pricing for several well-known unweighted networks with  uniform intrinsic values:
    $(\ln n)$-approximation for general networks, 1.5-approximation for forest networks,
    $(1+\epsilon)$-approximation a.a.s for {Erd\H{o}s}-R\'enyi networks, and
     2-approximation a.a.s. for Barab\'asi-Albert networks (a.k.a. preferential attachment networks).
\end{itemize}
{This justifies the importance of the research of both iterative pricing and single pricing, whose relations in various scenarios  represent different {trade-offs} between revenue efficiency and algorithmic simplicity.}

\paragraph{Related work.} In the economics literature, the importance of network effects and network externalities in business began to attract serious attention around three decades ago (\cite{farrell1985,katz1985network}). Under the most popular frameworks,  network effects are assumed to be global instead of local. {Namely,} 
only complete networks are considered. Consumers may also act sequentially as in this paper, but are usually assumed to be completely rational in the way that they are able to forecast the decisions of later ones and make their purchasing decisions accordingly. There are quite a lot of followups, most of  which are beyond the scope of this paper. We refer the reader to \cite{radner2014dynamic} for a most recent development in this paradigm with relaxations of assumptions on consumers.

In the literature of computer science, network pricing stems mainly from the study of diffusion and cascading. {One of the most important differences between this strand of research and that of economics is arguably} that {network  structures are explicitly} and seriously addressed.} Over the last decade, under the framework of viral marketing, the algorithmic study of diffusing products with positive
externalities is especially fruitful for influence maximization, {see, e.g., \cite{c09,kkt03,mr07}.
{To the best of our knowledge}, Hartline et al. {\cite{hms08}} {was the first to study} the diffusion problem} from a network pricing perspective.
They investigated marketing strategies for revenue
maximization with positive externalities. {Consumers are visited in a sequence (determined by the seller), and asked whether to buy or not under some
price (different consumers may receive different prices, referred to as differential pricing or discriminative pricing)}.  They {showed that for myopic consumers,} a reasonable approximation of the optimal marketing strategy can be achieved in a simple way of influence-and-exploit. While
complete information was assumed in \cite{hms08}, {Chen et al.} \cite{clstwz11} studied the incomplete information model with rational players and positive externalities. 
They  provided a polynomial
time algorithm that computes all the pessimistic (and optimistic) equilibria {and} 
the optimal single price. When discriminative pricing is
allowed, they proved the NP-hardness {of optimal equilibrium computation}, and gave an FPTAS  for the case that consumers are already partitioned into groups such that those within
the same group must receive the same price. 

Iterative pricing,  with a very limited literature, was   discussed {by Akhlaghpour et al. \cite{Akhlaghpour10} for positive
{externalities}. The authors studied} two iterative pricing models in which consumers are assumed to be {myopic}. In the {first model}, 
they {gave an} FPTAS for the optimal pricing strategy in the general case. In {the} {second model}, 
they {showed} that the revenue maximization problem is inapproximable even {in} some special case. Their second model is quite similar to {ours}.   

Although there is also a large literature in the field of classical
economics studying negative externalities ({under various terms,} e.g. the Veblen effect, the snob effect, the congestion effect etc.), explicit networks are rarely treated
seriously {as aforementioned}.   One of the classical papers in this strand \redcomment{is} \cite{jehiel1996not}, where the nulclear weapon selling problem was considered from the perspective of network effects. In the more recent computer science literature, compared with positive externalities, network pricing problems with negative externalities
are much less investigated. {Chen et al.} \cite{clstwz11} showed that when both positive and negative externalities are   allowed in their model, computing any approximate
equilibrium is PPAD-hard. 
However, the complexity status of the problem in the case with only negative externalities is still unknown. 
The only paper known to us that deals with the network pricing problem with negative externalities is
\cite{Bhattacharya2011} {by Bhattacharya et al.}, although their main focus is on equilibrium computation for given prices rather than pricing. The authors also considered
linear externalities, but a combination of single pricing, complete information {and strategic consumers. 
They
showed that for any given price, the game that the consumers play is an exact potential game, and provided a set of hardness results. They
proved that finding the best equilibrium is NP-hard even for trees, and gave a 2-approximation algorithm for bipartite networks. Along a different line, Alon et al. \cite{amt13} {used} the term {``negative externality'' to} 
mean the harm of discriminative pricing on consumers
(because discriminative pricing gives many consumers a feeling of inequality).

All the papers cited above assume that externalities are  only exerted  between consumers who buy the product. In contrast, for some products or sevices, e.g., public goods, 
externalities are exerted  from purchasers to nonpurchasers. {Our paper is close to \cite{bk07} in the sense that both papers address
strategic substitutes (each player has less incentive to buy when more neighbors
purchase), although the network externalities are negative in our settings but positive
in their {settings of public goods}.} {In the computer science}, the public goods pricing problem was also studied by Feldman et al.~\cite{fklp13}.
Their work differs from ours in two main respects: (i) {In} our externality model, {a consumer's} utility is subtractive over the purchases made by this neighbors, whereas in their setting, purchases of neighbors are substitutes. (ii) Technically, they {related} the pricing problem (where externalities in their model are mathematically expressed in terms of products of neighbors actions) to a single-item auction problem, while we address the pricing problem (where externalities are expressed in terms of sums of neighbors actions) using iterative algorithmic approaches. {As noted by the authors \cite{fklp13},  their results  carry over to a
special kind of negative externality, where the valuation of a consumer
on the product is positive if and only if the consumer is the only one among her/his neighbors who possess the product.} {The {aforementioned} literature  are all on indivisible goods. The {network} pricing problems for  divisible goods with quadratic utilities functions have been studied in \cite{bq13,cbo12}.} Along with \cite{fklp13}, {a growing number of papers have been addressing} the network externality problem from the perspective of mechanism design and auction theory (e.g. \cite{bateni2013revenue,deng2011money,haghpanah2013optimal}).

\bigskip {The remainder of the paper is organized as follows. Section \ref{sec:model} gives the mathematical formulation of our iterative pricing model. Section \ref{sec:ip} is devoted to general iterative pricing, including NP-hardness (Section~\ref{sec:np}), 2-approximation for general weighted network with general intrinsic values (Section~\ref{sec:2apx}) and optimal pricing for unweighted split network with uniform intrinsic values (Section \ref{split}). Section \ref{sec:single} discusses the relation between single pricing and iterative pricing. Single pricing is shown to guarantee 1.5-approximation for forests (Section \ref{appforest}), near optimal for  Erd\H{o}s-R\'enyi networks (Section \ref{sec:er}),  $(2-\epsilon)$-approximation for Barab\'asi-Albert networks (Section \ref{sec:ba}), {and}  approximation with ratio within $[\ln\ln n,\ln n]$  for general networks (Section \ref{upplow}). Section \ref{sec:conclude} concludes the paper with remarks on future research.}

\section{{The model}}\label{sec:model}

 {Let $G=(V,E)$ be the given undirected {network} (without {self-loops}, and possibly associated with a nonnegative {integer} weight
function $w\in\mathbb Z_+^{V\times V}$),  where $V\equiv [n]$ is the set of $n$ consumers, {and} $E$ represents the   links between pairs of consumers.
When the weight function $w\in\mathbb Z_+^{V\times V}$ is discussed, it is always assumed that {$w_{ij}=w_{ji}$ for all $i,j\in V$} and $w_{ij}=0$ if and only if $ij\not\in E$.
Given any consumer $i\in V$ and subset $S\subseteq V$ of consumers, we use $w_i(S)=\sum_{j\in S}w_{ij}$ to denote the sum of weights contributed to
consumer $i$ by those in~$S$. Clearly, only $i$'s neighbors can  {possibly} contribute.

{We name the  model {under investigation} as PNC ({\em Pricing with Negative externalities and Complete information}).
 Let $Q$, which {usually shrinks} 
 as the {iterative pricing} proceeds, denote the set of consumers
who do not possess the product. Each consumer $i\in V$ has an intrinsic value $\nu(i)\in\mathbb R_+$, and her  {\em total {value}} of the product
equals $\nu(i)+w_i(Q)$. Initially $Q=V$. The PNC model proceeds as follows.}

{\begin{itemize}\item {\em Iterative pricing.} The monopolist seller announces prices $p_1,p_2,\dots,p_{\tau}$ sequentially at time
$1,2,\ldots,\tau$.

\vspace{-1mm}\item {\em Impatient consumers.} As soon as a price {is} announced, a consumer in $Q$ buys the product if and only if her {current} total {value} is greater than or
equal to the current price.
\vspace{-1mm}\item {\em Simultaneous {moves}.} We assume that,  for each newly announced price,  all consumers in $Q$  make their decisions (buying or not buying) simultaneously.
\end{itemize}}

Note that a consumer in $Q$ who does not purchase at current time $t$ under price $p_t$ may be willing to buy at a later time $t'>t$
under a lower price $p_{t'}<p_t$. For each $t=1,2,\ldots,\tau$, let $B(p_t)$ denote the set of consumers who buy the product at price $p_t$, {(i.e., at time $t$, or in the $t$-th {\em round})}. We use
$\rv(\mathbf p)$ to denote the revenue derived from $\mathbf p=(p_1,p_2,\ldots,p_{\tau})$, {i.e.,} 
 $\rv(\mathbf p)=\sum_{t=1}^{\tau}p_t\cdot|B(p_t)|$. {In case of $\mathbf p=(p_1)$, we often write $\rv(\mathbf p)$ as $\rv(p_1)$.}
The PNC problem is to find a pricing sequence $\mathbf p=(p_1,p_2,\ldots,p_{\tau})$ such that $\rv(\mathbf p)$ is maximized, where both {the length $\tau$  and the entries
$p_1,p_2,\ldots,p_{\tau}$ of the sequence} are variables to be determined.

\section{{General iterative pricing}}\label{sec:ip}

{In this section, we study the PNC model in the most general setting where no restriction is imposed to the length  of the pricing sequence.}
\subsection{{NP-hardness}}\label{sec:np}

{We prove that finding an optimal pricing sequence for the PNC model is NP-hard, even when the intrinsic values are all zero, link weights are unit, and the network is a tree. Throughout this subsection, we assume that the intrinsic values of all consumers are zero.}

{We begin with some preliminaries that will be used in the formal proofs. Let  $\mathbf{p}=(p_1,p_2,\ldots,p_{\tau})$ be a pricing sequence.}  For any $t\in [\tau]$ and $i\in V$, we use $\nu_t(i,\mathbf p)$ to denote the {(total)} value of the product at time $t$ {(in the $t$-th round)  for
consumer $i$ during the selling/purchase process.  Since intrinsic value $\nu(i)$ is zero by assumption,} 
$\nu_t(i,\mathbf p)$ is the sum of weights from $i$'s neighbors who have not purchased yet in the previous rounds. 

\begin{observation}During the selling process, the value of the product for each consumer $i$ does not increase, i.e. $\nu_{t+1}(i,\mathbf p)\leq \nu_t(i,\mathbf p)$ {for all $t=1,2,\ldots,\tau-1$.}\end{observation}

Given 
a subset of nodes $S\subseteq V$, we use $\rv_{t,S}(\mathbf p)$ to denote the
revenue from these consumers until time $t$, {i.e., $\rv_{t,S}(\mathbf p)=\sum_{i=1}^{t}|S\cap B(p_i)|p_i$.} For brevity, we also write $\rv_{\tau,S}(\mathbf p)$ as $\rv_{S}(\mathbf p)$.  {In particular, we have $\rv_V(\mathbf p)=\rv(\mathbf p)$.}

\begin{definition}  {We call  pricing sequence $\mathbf{p}=(p_1,p_2,\ldots,p_{\tau})$ {\em irredundant} if for each
$i\in[\tau]$, there is at least one consumer who purchases under   price $p_{i}$}.\end{definition}

{Every pricing sequence $\mathbf{p}$ is ``equivalent'' to a unique irredundant pricing sequence $\mathbf{p}'$ which is derived from $\mathbf{p}$ by removing all prices under which no consumers purchase. Clearly, the equivalent pricing sequences bring  about the same revenue $\rv(\mathbf{p})=\rv(\mathbf{p}')$. This allows us to focus on irredundant pricing sequences.}

\begin{observation} \label{distinct}If pricing sequence $\mathbf{p}=(p_1,p_2,\cdots,p_{\tau})$ is irredundant, {then}  
it is decreasing, i.e. $p_1>p_2>\cdots>p_{\tau}$.
\end{observation}

Assume that $\mathbf{p}=(p_1,p_2,\ldots, p_{\tau})$ is an irredundant pricing sequence.  Since by Observation \ref{distinct} the entries of $\mathbf{p}$
are all distinct, we also {view} $\mathbf{p}$ as a set {$\{p_1,p_2,\ldots,p_{\tau}\}$}, and use the symbol  $p_i\in \mathbf{p}$ to mean that {$p_i$ is an entry of}  
the pricing sequence $\mathbf p$.
 For all $1\leq t\leq \tau$, {define $B_t(\mathbf{p})=\cup_{i=1}^{t}B(p_i)$} 
to be the set of
consumers who have purchased in the first $t$ rounds. For {notational convenience, we set  $B_0(\mathbf{p})=B(p_0)=\emptyset$.} Recall that in the PNC model}  we have assumed that consumers are all impatient in the sense that they will definitely purchase as long as
the current price is lower than or equal to their current values. {As $\mathbf p$ is irredundant,  $B(p_1),B(p_2),\ldots,B(p_{\tau})$} can be computed in a recursive
  way: {$B(p_t)=\{i\in V: w_i(V\setminus B_{t-1}(\mathbf p))\geq p_t\}$, $t=1,2,\ldots,\tau$.}

\begin{definition} {A pricing sequence $\mathbf{p}=(p_1,p_2,\ldots,p_{\tau})$ is called {\em normal} if it is irredundant, and for any $i\in[\tau]$ and any $\epsilon>0$,
 increasing $p_i$ to $p_i+\epsilon$   (other prices {remain} the same) changes the set of consumers who purchase at the $i$-th round.}
\end{definition}

{Clearly, all entries of a normal pricing sequence are integers. Given an irredundant pricing sequence $\mathbf{p}=(p_1,p_2,\ldots, p_{\tau})$ together with $B(p_1),B(p_2),\ldots,B(p_{\tau})$, one can easily compute a normal pricing sequence $\mathbf{p}'=(p'_1,p'_2,\ldots, p'_{\tau})$, which is ``equivalent'' to $\mathbf p$ in the sense that $B(p_t')=B(p_t)$ for all $t\in[\tau]$, as follows: $p'_t=\min\{w_i(V\setminus B_{t-1}(\mathbf p)):i\in B(p_t)\}$, $t=1,2,\ldots,\tau$. It is clear that $\rv(\mathbf p')\ge\rv(\mathbf p)$. The following observation enables us to concentrate on normal pricing sequences in our NP-hardness proofs.}

\begin{observation}There is an optimal pricing sequence that is normal. 
\label{o3}
\end{observation}

{The NP-hardness for the PNC model is proved  by reduction from the 3SAT problem. The input of the 3SAT problem are $n$ boolean variables $x_1,x_2,\cdots,x_n$, and $m$ clauses
{$c^j=(x^{j1}\vee x^{j2}\vee x^{j3})$}, $ 1\leq j\leq m$,} where $x^{j\ell}$ is a literal taken from $\{x_1,x_2,\cdots,x_n,\neg x_1,\neg
x_2,$  {$\cdots, \neg x_n\}$}, $1\leq
j\leq m, 1\leq \ell\leq 3$. For convenience, we write $x^{j\ell}\in c^j$. {The 3SAT problem is to determine if there is a satisfactory truth assignment to the $n$ variables that makes all $m$ clauses evaluate to TRUE.} 
To avoid triviality, we assume that $m\ge3$, and for each $i\in[n]$, there exist $j,j'\in[m]$ such that $x_i\in c^j$ and $\neg x_i\in c^{j'}$.

{Next, we prove the NP-hardness for   the weighted case with general network structures. The proof, which highlights the high level idea in our later proof to handle the unweighted case with tree structures, turns out to be much easier to understand.}

\begin{theorem}\label{easynp}
In the PNC model, computing {an}  optimal pricing sequence is NP-hard, even when all the intrinsic values are zero. 
\end{theorem}
\begin{proof}
{Our reduction here uses} a slightly restricted version of the 3SAT problem,  {the 3-OCC-3SAT problem},  which is known to be NP-hard, where for each $i\in[n]$, there are at most three clauses that contains either $x_i$ or $\neg x_i$.

For any instance $I$ of {the 3-OCC-3SAT problem}, we construct an instance $P$ of the network pricing problem  {on network $G=(V,E)$} as follows. There are a total of $5n+3m$
nodes:\begin{itemize}\item For each variable $x_i$, there is a gadget $V_i$. Each pair of literals $x_i$ and $\neg x_i$ are simulated by two nodes
(with {the  names unchanged}), respectively, and three auxiliary ones, $y_{i1},y_{i2},y_{i3}$. See the left part of Figure~1 for the links and {their weights.}
\vspace{-1mm}\item For each clause $c^j=(x^{j1}\vee x^{j2}\vee x^{j3})$, there is a gadget $C^j$.
The clause is simulated by a node $c^j$ and two auxiliary ones, $d^j$ and $e^j$. See the right part of Figure 1 for the links and weights among them.

\vspace{-1mm}\item A literal node ($x_i$ or $\neg x_i$) is linked to a clause node $c^j$ if and only if this literal appears in the clause, and the weight {of the link} is 1.

\vspace{-1mm}\item {The integer parameters} in the weights satisfy
{\begin{equation}a_1>5 a_2> \cdots >5^{i-1}a_i>\cdots>5^{n-1}a_n>5^n a>5^{n+1} mn.\label{e01}\end{equation}}\end{itemize}

\begin{figure}[h]
\begin{center}
\includegraphics[width=10cm]{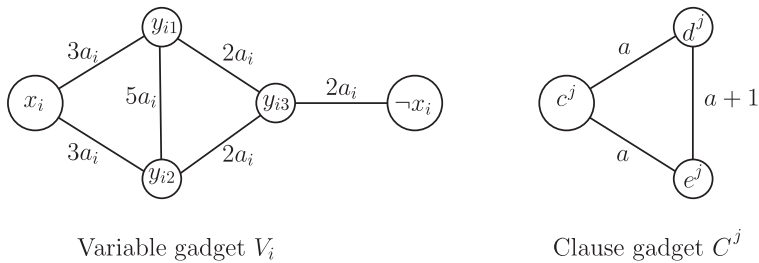}
\caption{
Literal nodes and clause nodes are represented by larger circles, while
auxiliary nodes are represented by smaller ones.} \label{Figure.1}
\end{center}\end{figure}

Obviously, the above construction can be done in polynomial time. Observe first that all the consumers in the variable gadgets {are incident {with} links of weights much larger than the total weight of links that are incident with any clause consumer.} 
This structure permits us to consider the variable consumers before the clause ones.
In the rest of this proof, we may
abuse the notations $V_i$ and $C^j$ a little bit to represent both the gadgets and the corresponding node sets, respectively.

 {Due to Observation \ref{o3}, we only consider normal pricing sequences.} Given any  {normal} pricing sequence $\mathbf{p}=(p_1,p_2,\ldots, p_{\tau})$, let {$\xi$} be the first time that the price is equal to or lower than $2a+3$, i.e.,
\begin{equation*}\xi=\min\{t: p_t\in \mathbf{p},p_t\leq 2a+3\}.\end{equation*}

Note that before time $\xi$, no consumer in the clause gadgets has purchased, i.e., \begin{equation*}B_{\xi-1}(\mathbf p) \cap (\cup_{j=1}^m
C^j)=\emptyset.\end{equation*}

The key idea of our proof is simple: we shall show that for each pair of nodes $x_i$ and $\neg x_i$, we can sell the product to one and only one of them, and this makes no difference {for the revenue} at all before time $\xi$ {(see Claim \ref{c01} below)}. The only difference that the choice between $x_i$ and $\neg x_i$ makes is upon the clause gadget nodes after time $\xi$. Our construction makes these choices really hard because they correspond to a (possible) solution of the 3-OCC-3SAT problem.

For any $i\in[n]$, we note that {$w_v(V)\le10a_i$} for all $v\in V_i$; thus no consumer in $V_i$ purchases when the price is above $10a_i$. 
\begin{clm}\label{c01}\label{c02}\label{c03}
For all $i=1,2,\ldots,n$,
 \begin{mylabel}
 \item if $\mathbf{p}\cap [2a_i,10a_i]=\{10a_i,2a_i\}$, then $\rv_{\xi-1,V_i}(\mathbf p)=24a_i$, $x_i\notin B_{\xi-1}(\mathbf p) $ and $\neg x_i\in B_{\xi-1}(\mathbf p) $;


\item if $\mathbf{p}\cap [2a_i,10a_i]=\{6a_i\}$, then $\rv_{\xi-1,V_i}(\mathbf p)=24a_i$, $x_i\in B_{\xi-1}(\mathbf p) $ and $\neg x_i\notin B_{\xi-1}(\mathbf p) $;


\item $\rv_{\xi-1,V_i}(\mathbf p)\leq 24a_i$, and the equality holds if and only if $\mathbf{p}\cap [2a_i,10a_i]\in \{\{10a_i,2a_i\},\{6a_i\}\}$.\end{mylabel}\end{clm}
{Statements (i) and (ii) are easily checked. It remains}  to prove 
$\rv_{\xi-1,V_i}(\mathbf p)<24a_i$ {if} $\mathbf{p}\cap [2a_i,10a_i]\notin \{\{10a_i,2a_i\},\{6a_i\}\}$. Note that $\rv_{\xi-1,V_i}(\mathbf p)<24a_i$ is trivial if $\mathbf{p}\cap [2a_i,10a_i]=\emptyset$. Hence we may assume that there exists a maximum price  $\hat{p}\in \mathbf{p}\cap[2a_i,10a_i]$. 
By {normality} of $\mathbf{p}$, we know that $\hat{p}\in \{10a_i,6a_i+h_i,6a_i,2a_i+h'_i,2a_i\}$, where {$h_i$ and $h_i'$ are total weights that $x_i$ and $\neg x_i$ get from clause gadgets, respectively. Hence} $h_i+h_i'\le3$ {(recall the definition of 3-OCC-3SAT)}. For the case that $\hat{p}\le2a_i+h'_i$, it is obvious that
$\rv_{\xi-1,V_i}(\mathbf p)<24a_i$. We are left to the analysis of the remaining three cases, which will {establish Statement~(iii).} \begin{itemize}
\item $\hat{p}=10a_i$. It follows from $\mathbf{p}\cap [2a_i,10a_i]\neq\{2a_i,10a_i\}$  that $\mathbf{p}\cap [2a_i,10a_i]=\{10a_i\}$,  because the only price in $[2a_i,10a_i]$ that is smaller than $10a_i$ and  makes $\mathbf{p}$ normal is $2a_i$.  This gives $\rv_{\xi-1,V_i}(\mathbf p)=20a_i<24a_i$.

\vspace{-1mm}\item $\hat{p}=6a_i+h_i$. The
normality of $\mathbf{p}$ implies {$\mathbf{p}\cap [2a_i,10a_i]\in \{\{6a_i+h_i,2a_i+h'_i\},\{6a_i+h_i,2a_i\},\{6a_i+h_i\}\}$} and  {hence} $\rv_{\xi-1,V_i}(\mathbf p)\le3(6a_i+h_i)+4a_i<24a_i$.
\vspace{-1mm}\item {$\hat{p}=6a_i$}. An argument similar to  {the previous case shows that} $\rv_{\xi-1,V_i}(\mathbf p)\le3\times6a_i+4a_i<24a_i$.\qedd\end{itemize}

\begin{clm}For each $1\leq j\leq m$, $\rv_{C^j}(\mathbf p)\leq 6a+3$, and the {equality} holds if and only if
$D^j \setminus B_{\xi-1}(\mathbf p) \neq \emptyset$ and $\mathbf{p}\cap [2a,2a+3]=\{2a+1\}$, where
$D^j=\{x^{j1},x^{j2},x^{j3}\}$.\label{c04}\end{clm}

It is {easy to check} that when {$D^j \setminus B_{\xi-1}(\mathbf p) \neq \emptyset$} and
$\mathbf{p}\cap [2a,2a+3]=\{2a+1\}$, the equality $\rv_{C^j}(\mathbf p)=6a+3$ holds. {We prove}   $\rv_{C^j}(\mathbf p)<6a+3$ in the other cases. When
{$D^j \subseteq B_{\xi-1}(\mathbf p)$ or $\mathbf{p}\cap
(2a,2a+3]=\emptyset$}, it is easy to see that  $\rv_{C^j}(\mathbf p)$ is at most $2a\times 3=6a$. So we {only need} to discuss the case of $D^j \setminus B_{\xi-1}(\mathbf p) \neq \emptyset$   {and $\mathbf{p}\cap
(2a,2a+3]\ne\emptyset$}.

{Let $\hat{p}\in \mathbf{p}\cap(2a,2a+3]$ be maximum.} 
 By {normality} of $\mathbf{p}$, we know that
{$\hat{p}\in \{2a+1,2a+2,2a+3\}$.} Since $w(d^j)=w(e^j)=2a+1$, it can be seen that
$\rv_{C^j}(\mathbf p)<6a+3$ holds for $\hat{p}\in \{2a,2a+2,2a+3\}$. 
So {Claim \ref{c04}} is valid. \qedd

\medskip We are now ready to prove the close relation between {the 3-OCC-3SAT instance $I$ and the PNC instance~$P$.} Let {$\opt(P)$} be the optimal objective value of $P$. Define
\begin{equation*}L=\sum_{i=1}^{n}24a_i+m(6a+3).\end{equation*}

\begin{clm}$\opt(P)\leq L$.\label{c05}\end{clm}

{Suppose that $\mathbf p$ is an optimal solution of $P$.} {We can assume without loss of generality that $p_\tau>0$. Let $U=V\setminus B_{\tau}(\mathbf p)$ denote the set of consumers who do not purchase during the whole selling process.} Note first from our previous discussion that $\rv_{V_i}(\mathbf p)$ may be greater than $24a_i$, although {$\rv_{C^j}(\mathbf p)\leq 6a+3$} holds for every $j\in [m]$. However, $\rv_{V_i}(\mathbf p)\leq 24a_i+3$ is always valid, because {$\sum_{v\in V_i}w_v(\cup_{j=1}^mC^j)\le3$} (recall  the definition of 3-OCC-3SAT {and the construction of $P$}). Also, when {$V_i\setminus U\subseteq B_{\xi-1}(\mathbf p)$, we do have} $\rv_{V_i}(\mathbf p)\leq 24a_i$. Therefore, if {$(\cup_{i=1}^n V_i)\setminus U\subseteq B_{\xi-1}(\mathbf p)$}, the above claim is derived immediately from Claims \ref{c03}  and \ref{c04}.

Suppose some {$v\in V_i\setminus U$ purchases at price $p_t>0$} with $t\ge \xi$. If $v\in \{y_{i1},y_{i2},y_{i3}\}$, or {$v\in\{x_i,\neg x_i\}$} and $p_t>3$,
 then it can be seen easily that  $\rv_{\xi-1,V_i}(\mathbf p)\le20a_i$, and {hence (\ref{e01}) implies $\rv_{V_i}(\mathbf p)\le20a_i+5(2a+3)<20a_i+2a_i+15<24a_i$.} {It remains to consider the case where} $\{y_{i1},y_{i2},y_{i3}\}\subseteq B_{\xi-1}(\mathbf p)$, {$0<p_t\le3$,} and $v\in \{x_i,\neg x_i\}\cap D^{j_0}$ for some $j_0\in[m]$ {with $c^{j_0}\not\in B_{t-1}(\mathbf p)$}. {Since $c^{j_0}$ does not purchase before time $t$}, it must be the case that $\mathbf{p}\cap(p_t,2a+1]=\emptyset$. 
 It follows {$p_t\le3$} that $\rv_{C^j}(\mathbf p)\le 9$ for all $j\in [m]$. Hence (\ref{e01}) implies $\rv(\mathbf p)=\rv_V(\mathbf p)\le \sum_{i=1}^{n}(24a_i+3)+9m<L$. 
So Claim \ref{c05} is indeed correct.\qedd

\medskip To establish the NP-hardness of the pricing problem, it suffices to prove that
\begin{equation*}\opt(P)\geq L \Leftrightarrow I ~\mbox{is satisfiable}. \end{equation*}

($\Leftarrow$) Suppose that $I$ has a satisfactory truth assignment $\pi$ with $s$ variables assigned ``TRUE" and the {remaining} $n-s$ variables  assigned
``FALSE". Let {pricing sequence} $\mathbf{p}=(p_1,p_2,\ldots,p_{n+s+1})$ be a solution to $P$ such that
\begin{itemize}
\item There are one or two prices for each variable gadget {depending on} whether the variable is assigned ``TRUE'' or ``FALSE'' in $\pi$:   if $x_{i}$ is assigned ``TRUE"  then $  10a_i,2a_i\in\mathbf p$,   if $x_{i}$ is
assigned ``FALSE"  then $  6a_i\in\mathbf p$;
\vspace{-1mm}\item  There is  a common price for the $m$ clause gadgets: $p_{n+s+1}=2a+1\in\mathbf p$.
\end{itemize}
For this $\mathbf p$,  note {from (\ref{e01})} that $\xi=n+s+1$. According to Claim \ref{c03}, {$\rv_{\xi-1,V_i}(\mathbf p)=24a_i$ for each $i\in [n]$}. For each clause gadget $C^j$, due to Claim
\ref{c01}(i) and (ii), we know that consumer $x^{j\ell}\in B_{\xi-1}(\mathbf p)$, $1\leq \ell\leq 3$, if and only if the corresponding literal  is ``FALSE" in $\pi$. Since $\pi$ is a satisfactory assignment, we know that there is at least one literal in $x^{j1},x^{j2},x^{j3}$ that is assigned
true. Therefore, for each $j$, it holds that $D^j \setminus B_{\xi-1}(\mathbf p)\neq \emptyset$. Combining this fact with Claim \ref{c04}, we know
that $\rv_{\xi,C^j}(\mathbf p)=6a+3$ for each clause $C^j$. This completes the sufficiency part.

($\Rightarrow$) Suppose now $\opt(P)\geq L $. Due to Observation \ref{o3}, there exists a {normal} pricing sequence $\mathbf{p}=(p_1,p_2,\ldots,p_{\tau})$
whose objective {value $\rv(\mathbf p)$} is at least $L$. Combining with Claim \ref{c05}, this can only be the case that $\opt(P)=L$. By arguments in the proof of Claim \ref{c05}, we know that conditions in Claims \ref{c03} and   \ref{c04} must hold. We construct a truth assignment $\pi$ as follows: for each
{$i\in[n]$}, if $\mathbf{p}\cap [2a_i,10a_i]=\{2a_i,10a_i\}$, we assign ``TRUE" to variable $x_i$. Otherwise, that is $\mathbf{p}\cap
[2a_i,10a_i]=\{6a_i\}$, {we assign ``FALSE" to   $x_i$}. By Claim \ref{c04}, we know that $D^j \setminus B_{\xi-1}(\mathbf p) \neq \emptyset$ for all {$j\in[m]$}. Therefore $\pi$ is indeed a satisfactory truth assignment for $I$. This completes the necessity part, {and therefore the proof of Theorem \ref{easynp}}.
\end{proof}

{A corollary of the above proof says that the length of the optimal pricing sequence of the PNC problem can not be upper bounded by any constant. This remains true for the unweighted trees without intrinsic values, as seen from the proof of the following stronger NP-hardness result.}


\begin{theorem}\label{hard}{
{In} the PNC model, computing {an}  optimal pricing sequence is NP-hard, even when the underlying network is an unweighted tree and all the intrinsic
values are zero.} \qed
\end{theorem}

{While the proof,  which we postpone  to the appendix, has high level similarities to the one for Theorem~\ref{easynp}, a substantially more careful approach is required to handle the acyclic structure, and new ideas are needed to simulate the weights with unweighted links.}

In view of the above NP-hardness result, it is {desirable} to design good approximation algorithms for the general {PNC} problem and exact
algorithms for special cases. {In the following, we obtain 2-approximation for the general case (Theorem \ref{2-app}), and an optimal pricing for unweighted split networks (Theorem \ref{th:split}).}

 \subsection{{$2$-approximation}}\label{sec:2apx}

As to approximation, we find that, more or less surprisingly, a very simple greedy algorithm performs fairly well, {achieving 2-approximation} for the most general scenario.

For any  {subnetwork} $H$ of $G$, and any $i\in V(H)$, where $V(H)$ is the node set of $H$,  we use
$d^w_H(i)=\sum_{j\in V(H)}w_{ij}$ to denote the weighted degree of $i$ in $H$. {For any real function $f$ and any nonempty subset $S$  of its domain, let $f(S)=\sum_{s\in S}f(s)$.} 

\begin{algorithm} 
\KwIn{Network $G=(V,E)$
with weight {function} $w\in\mathbb Z_+^{V\times V}$ and intrinsic value function $\nu\in\mathbb R_+^V$.} \KwOut{Sequence $\mathbf p$ of prices.}
\begin{enumerate}
  \item[1.]   $G_0\leftarrow G$,\quad $t\leftarrow0$\label{init}
  \vspace{-1.6mm}\item[2.]  \textbf{While} $V(G_t)\neq \emptyset$ \textbf{do}\label{nonempty}
  \vspace{-1.6mm}\item[3.] \hspace{3mm} $t\leftarrow t+1$
  \vspace{-1.6mm}\item[4.] \hspace{3mm} $p_t\leftarrow \max\{\nu(i)+d^w_{G_{t-1}}(i):{i\in V(G_{t-1})}\}$\label{stp4}
  \vspace{-1.6mm}\item[5.] \hspace{3mm} $G_t\leftarrow G_{t-1}\setminus B(p_t)$
  \vspace{-1.6mm}\item[6.]  \textbf{End-while}
  \vspace{-1.6mm}\item[7.] Output $\mathbf p\leftarrow(p_1,p_2,\ldots,p_t)$
\vspace{-1.6mm}
\end{enumerate}
\caption{{Iterative Pricing}} \label{alg1}
\end{algorithm}

\begin{theorem}\label{2-app}{For the PNC model, Algorithm \ref{alg1} finds a $2$-approximate pricing sequence in $O(n^2)$ time}.
\end{theorem}
\begin{proof}
Let $\mathbf p^*$ be {an} optimal pricing, and $\mathbf p=(p_1,p_2,\ldots,p_{\tau})$ be the pricing output by the algorithm. Since each link $ij\in E$
can contribute at most $2w_{ij}$ to consumers' {total values} ($w_{ij}$ to each of $i$ and $j$), we see that \[\rv(\mathbf p^*)\le\nu(V)+2w(E).\]
On the other hand, the definition of $p_t$ in Step 4 of Algorithm \ref{alg1} guarantees that $B(p_t)={\text{argmax}}\{\nu(i)+d^w_{G_{t-1}}(i):i\in
V(G_{t-1})\}$ and therefore  each link $ij\in E(G_{t-1})$ with $i\in B(p_t)$ contributes $w_{ij}$ to $i$'s {total value}, giving
\[p_t\cdot|B(p_t)|=\nu(B(p_t))+\sum_{i\in B(p_t)}d^w_{G_{t-1}}(i)\ge\nu(B(p_t))+w(E(G_{t-1})-E(G_t)), \text{ for all } t\in [\tau].\]
It follows that
\begin{eqnarray*}
\rv(\mathbf p)=\sum_{t=1}^{\tau}p_t\cdot|B(p_t)|&\ge&\sum_{t=1}^{\tau}\nu(B(p_t))+\sum_{t=1}^{\tau}w(E(G_{t-1})-E(G_t))\\
&=&\nu(\cup_{t=1}^{\tau}B(p_t))+\sum_{t=1}^{\tau}w(E(G_{t-1}))-w(E(G_t))\\
&=&\nu(V)+w(E(G_0))-w(E(G_{\tau}))\\
&=&\nu(V)+w(E),
\end{eqnarray*}
where $G_0=G$ and $G_{\tau}=\emptyset$ are guaranteed by Steps 1 and 2. Hence \[\rv(\mathbf p^*)/\rv(\mathbf p)\le(\nu(V)+2w(E))/(\nu(V)+w(E))\le2\]
justifies the approximation ratio 2.

 To see the $O(n^2)$ running time, we note that the while-loop repeats $\tau\le n$ times, and each repetition finishes in $O(n)$ time.
\end{proof}

\subsection{Optimal pricing for unweighted split networks\label{split}}

 Network $G=(V,E)$ is a {\em split network} if  {its node set} $V$ can be partitioned into two sets $C$ and $I$ such that $C$ induces a clique   and $I$ is an
independent set of $G$. Clearly, the nodes in $I$ can  only have neighbors in $C$. In case of each node in $I$ adjacent to exactly one node in $C$,
network $G$ is called {\em core-peripheral}. Core-peripheral networks are widely accepted as good simplifications of many real-world networks and
thus have been extensively studied in various environments \cite{b07}.

{We consider the case of uniform intrinsic values,   which can be assumed w.l.o.g. to be zeros.}
Let {$d(v)=d_G(v)$} denote the degree of $v\in V$ in $G$. Suppose that $C=\{v_1,v_2,\ldots,v_k\}$, 
and $d(v_i)\leq d(v_{i+1})$ for every $i\in [k-1]$. For each $i\in [k]$, note that $ v_1,\ldots,v_i $ form a clique set $C_i$ and their neighbors in
$I$ form an independent set $I_i$, and $C_i\cup I_i$ induces a split subnetwork $G_i $ of $G$ with degree sequence
\[ d_{G_i}(u^i_{\ell_i})\le  d_{G_i}(u^i_{\ell_i-1})\le\cdots\le  d_{G_i}(u^i_{1}) \le d_{G_i}(v_1)\le\cdots\leq d_{G_i}(v_i),\]
where $I_i=\{u^i_1,u^i_2,\ldots,u^i_{\ell_i}\}$. Apparently, $d_{G_i}(v_h)=d(v_h)-(k-i)$ for every $h\in[i]$. Consider an optimal pricing $\mathbf
p=(p_1,\ldots, p_{\tau})$ for the {PNC} problem on  $G_i$, {and write the corresponding maximum revenue as $\opt(G_i)$}. One of the following must hold.
\begin{itemize}
\item   $p_1=d_{G_i}(v_{h+1})$  for some $h\in[i-1]$, and exactly $(i-h)$ nodes, i.e., $v_{h+1},\ldots,v_i$, purchase at price $p_1$, offering revenue   $(i-h)p_1=(i-h)d_{G_i}(v_{h+1})$. It follows that $\tau\ge2$ and $ (p_2,\ldots,p_{\tau})$ is an optimal  pricing for $G_{h}$, giving $\opt(G_i)=(i-h)d_{G_i}(v_{h+1})+\opt(G_h)=(i-h)d (v_{h+1}-k+i)+\opt(G_h)$.
\vspace{-1mm}\item   $p_1=d_{G_i}(u^i_j)$  for some $j\in [\ell_i]$ and exactly $(i+j)$ nodes, i.e., $u^i_j,u^i_{j-1},\ldots,u^i_{1}$, $v_1,v_2,\ldots,v_i$, purchase at price $p_1$, offering revenue is $(i+j)p_1=(i+j)d_{G_i}(u^i_j)$. Since  the nodes   not purchasing at price $p_1$ are pairwise {nonadjacent}, it is easy to see that $\mathbf p=(p_1)$ and $\opt(G_i)=(i+j)\cdot d_{G_i}(u^i_j)$.
\end{itemize}
For convenience, let $\opt(G_0)$ stands for {real number} $0$. Then $\opt(G)=\opt(G_k)$ can be computed by the following recursive {formula}:
\begin{eqnarray*}
\opt(G_i)=\max\left\{\max_{h=0}^{i-1}\{\opt(G_{h})+(d(v_{h+1})-k+i)(i-h)\},\max_{j=1}^{\ell_i}\{(j+i)\cdot d_{G_i}(u^i_j)\}\right\}\text{ for }i=1,2,\ldots,k.
\end{eqnarray*}
This formula  {implies the following result.}

\begin{theorem} \label{th:split}{For the PNC model, an optimal pricing sequence for any unweighted split network with uniform intrinsic values can be found in $O(n^2)$ time by dynamic programming.} \qed\end{theorem}

\section{{Approximation by single pricing}
}\label{sec:single}

{Finding an optimal single pricing is trivial because it can be chosen from the $n$ total values of the consumers. Thus it  is natural to ask: How
does the optimal single pricing work as an approximation to the optimal iterative pricing?}   We find that the answer is both ``good" and ``bad", in
the sense that single pricing {works} rather well for many interesting {networks with unit weights and uniform intrinsic values, including forests, Erd\H{o}s-R\'enyi networks and Barab\'asi-Albert networks}, but  {in general, its worst-case performance, even when restricted to unweighted networks, can be} arbitrarily bad. This justifies the
importance of the research of iterative pricing, and at the same time poses the interesting question of investigating the relation between single
pricing and iterative pricing for more realistic scenarios.

{In this section, we restrict our attention to unweighted networks $G$ with uniform intrinsic values, for which we may assume without loss of generality that all intrinsic values are zero, and use $\opt(G)$ to denote the revenue derived from an optimal iterative pricing.}

\subsection{$1.5$-approximation for {forests}}\label{appforest}
{We show that the {best} single price guarantees an approximate ratio of 1.5 for unweighted forests with uniform intrinsic values.}
\begin{theorem}\label{forest}{For the PNC model, the} single pricing $p$ with maximum  $p\cdot|B(p)|$ {has an approximation ratio of $1.5$}  for  {unweighted forests with uniform intrinsic values}. 
\end{theorem}

\begin{proof} 
Suppose that forest $G=G_0$ consists of $k$ components (trees) $G_h=(V_h,E_h)$, $h=1,\ldots,k$.  {Let $\ell_h$ denote the number of leaves in $G_h$. 
Note that} $\sum_{i=1}^k\opt(G_i)\ge \opt(G_0)$, and  \begin{eqnarray*} p\cdot|B(p)|
\ge\max\{|B(1)|,2|B(2)|\}= \max\{|V_0|, 2(|V_0|-\ell_0)\}\ge \frac23\left(2|V_0|-\ell_0\right)=\frac23\left(2\sum_{i=1}^k|V_i|-\sum_{i=1}^k\ell_i\right).\end{eqnarray*}
 {It} suffices to show that   $
\opt(G_i)\le 2|V_i|-\ell_i$ for each $i\in[k]$, {in order to guarantee that the approximation is at most 1.5}.

If $G_i$ is a star network  {or a link}, then   $\opt(G_i)=|V_i|\le 2|V_i|-\ell_i$. Suppose that $G_i$ is {neither a star nor a link}, and let $G_i' $ be the tree obtained from $G_i$ by deleting all its leaves.  
 {Clearly,  $ d_{G_i'}(v)\leq d_{G_i}(v)$   for every non-leaf node of $G_i'$.}

{Let $\mathbf p$ be an  optimal pricing for $G_i$. Consider an arbitrary leaf $u$ of $G_i'$. Let $L(u)$ denote {the set of $u$'s leaf neighbors in
$G_i$}. Under $\mathbf p$, either $u$ purchases before all nodes in $L(u)$ at a price higher than  {1} or all nodes in $\{u\}\cup L\{u\}$ purchase at
price {1}. {As $u$ has at least one non-leaf neighbor in $G_i$, it is easy to see that} in either case, the total payment by nodes in $\{u\}\cup L(u)$ is upper bounded by $d_{G_i}(u)\le d_{G_i'}(u)+|L(u)|$. Hence
\begin{eqnarray*}
\opt(G_i)&\le&\sum_{\text{non-leaf node $v$ of }G_i'}d_{G_i}(v)+\sum_{\text{leaf node $u$ of }G_i'}(d_{G_i'}(u)+|L(u)|)\\
&=&\sum_{\text{non-leaf node $v$ of }G_i'}d_{G_i'}(v)+\sum_{\text{leaf node $u$ of }G_i'}(d_{G_i'}(u)+|L(u)|)\\
&=&\ell_i+\sum_{v\in V_i}d_{G_i'}(v)\\
&=&\ell_i+ 2(|V_i|-\ell_i-1)\\
&<&2|V_i|-\ell_i,
\end{eqnarray*}
as desired.} \end{proof}
 \begin{remark}\label{tight1}
 {In Theorem \ref{forest}, to achieve  the approximation
ratio $1.5$, the single price can be simply chosen between 1 and 2, whichever produces a larger revenue. Moreover, the ratio $1.5$ is tight, as shown
by the following tree $G$.}
\end{remark}
{Tree $G$ with $n=1+2k$ nodes is a spider with center of degree $k$ and each leg of length 2 (i.e., the tree obtain from star $K_{1,k}$ by
subdividing each link with a node). It is easy to see that the maximum  revenue  $3k$ is given by pricing {sequence} $(k,1)$. However, any single pricing can
produce a revenue of at most $\max\{k\cdot 1,2\cdot (k+1),1\cdot (2k+1)\}=2k+2$. The tightness follows from $3k/(2k+2)\rightarrow 1.5$
($k\rightarrow\infty$).}

\subsection{Near optimal pricing for   {Erd\H{o}s-R\'enyi networks}} \label{sec:er}

For large $n$, there is a simple  algorithm that is ``almost optimal" for ``almost all'' Erd\H{o}s-R\'enyi networks $\mathbb G(n,\eta(n))$. {The network is}
constructed by connecting $n$ nodes randomly; each link is included in the network with
 probability $\eta(n)$. This algorithm, which will be referred to as  $A(\delta)$,  prices only once with  {price $(1-\delta) (n-1)\eta(n)$},
where $\delta>0$ is a parameter to be determined by the approximation ratio that we intend to reach. 

\begin{theorem}\label{th:er}Given arbitrarily small positive number $\epsilon>0$, set {$\delta\in(0,1)$} such that
\begin{equation}
\frac{1+\delta}{1-\delta}<1+\epsilon.\label{r1}
\end{equation} Then {for the PNC model,} Algorithm A($\delta$) has
 an approximation ratio of at most $1+\epsilon$ for asymptotically almost all networks $\mathbb G(n,\eta(n))$,
 as long as \begin{equation}
 \frac{\eta(n)}{\sqrt{\text{{$(\ln n)$}}/n}}\rightarrow +\infty.\label{r0}\end{equation}
To be precise, under condition (\ref{r0}), we
have\begin{equation}\lim_{n\rightarrow\infty}Pr\left(\frac{2|E(\mathbb{G}(n,\eta(n)))|}{r(\mathbb{G}(n,\eta(n)))}\leq
1+\epsilon\right)=1,\end{equation}
where $E(\mathbb G(n,\eta(n))$ is the link set of $\mathbb G(n,\eta(n))$, 
$Pr(\cdot)$ is {the
  probability function, 
and $r(\mathbb{G}(n),\eta(n))$ is the revenue obtained {from} the single pricing $(1-\delta) (n-1)\eta(n)$}.
\end{theorem}
\begin{proof}
Let $d_i$ be the degree of node $i$ in the random network $\mathbb{G}(n,\eta(n))$. As $0<\delta<1$, the following Chernoff bound holds:
\begin{equation}Pr(|d_i-(n-1)\eta(n)|>\delta(n-1)\eta(n))\leq 2\exp\left(-\frac{\delta^2\text{{$(\eta(n))^2$}}(n-1)}{2}\right).\label{r3}\end{equation}

Let $\alpha_n$ be the number of nodes in $\mathbb{G}(n,\eta(n))$ whose degrees fall  into $[(1-\delta)(n-1)\eta(n), (1+\delta)(n-1)\eta(n)]$. That
is, if we let $I_i$ be an indicator random variable such that $I_i=1$ if $d_i\in [(1-\delta)(n-1)\eta(n), (1+\delta)(n-1)\eta(n)]$ and $I_i=0$
otherwise, then
\begin{equation*}\alpha_n=\sum_{i=1}^{n}I_i.\end{equation*}

Now, we use $\alpha_n$ to bound $|E(\mathbb{G}(n,\eta(n)))|$ and $r(\mathbb{G}(n,\eta(n)))$ as follows: \begin{equation*}
2|E(\mathbb{G}(n,\eta(n)))|\leq \alpha_n (1+\delta)(n-1)\eta(n)+(n-\alpha_n)(n-1),\end{equation*}
\begin{equation*}r(\mathbb{G}(n,\eta(n)))\geq \alpha_n (1-\delta)(n-1)\eta(n).\end{equation*}
Therefore, \begin{equation*}\frac{2|E(\mathbb{G}(n,\eta(n)))|}{r(\mathbb{G}(n,\eta(n)))}\leq \frac{\alpha_n
(1+\delta)\eta(n)+(n-\alpha_n)}{\alpha_n (1-\delta)\eta(n)},\end{equation*}
 \begin{eqnarray*}Pr\left(\frac{2|E(\mathbb{G}(n,\eta(n)))|}{r(\mathbb{G}(n,\eta(n)))}\leq  1+\epsilon\right)
 &\geq&   Pr\left(\frac{\alpha_n
(1+\delta)p+(n-\alpha_n)}{\alpha_n (1-\delta)p}\leq 1
+\epsilon\right)\\&=&Pr\left(\alpha_n\geq n\epsilon_0\right),\end{eqnarray*} where $\epsilon_0=1/(
\epsilon(1-\delta)\eta(n)-2\delta \eta(n)+1)$, which is smaller than 1 due to (\ref{r1}). Using $\alpha_n=\sum_{i=1}^nI_i$, we have\begin{eqnarray*}
Pr\left(\frac{2|E(\mathbb{G}(n,\eta(n)))|}{r(\mathbb{G}(n,\eta(n)))}\leq  1+\epsilon\right) &\geq &
Pr\left(\sum_{i=1}^{n}\left(I_i-\epsilon_0\right)\geq 0\right)\\
&=&1-Pr\left(\sum_{i=1}^{n}\left(I_i-\epsilon_0\right)<0\right)\\
&\geq&1-Pr\left(I_i-\epsilon_0<0 \mbox{~holds for at some}~ i\in[n]\right)\\
&\geq&1-\sum_{i=1}^{n}Pr\left(I_i-\epsilon_0 <0\right)\\
&=&1-\sum_{i=1}^{n}Pr\left(I_i=0\right)\\
&=&1-\sum_{i=1}^{n}Pr\left(|d_i-(n-1)\eta(n)|>\delta (n-1)\eta(n)\right),
\end{eqnarray*} 
{where the second last equality is due to the fact that $I_i\in\{0,1\}$.} It follows from (\ref{r3}) and (\ref{r0}) that \begin{eqnarray*}
Pr\left(\frac{2|E(\mathbb{G}(n,\eta(n)))|}{r(\mathbb{G}(n,\eta(n)))}\leq  1+\epsilon\right)
\ge 1-\sum_{i=1}^{n}2\exp\left(-\frac{\delta^2p^2(n)(n-1)}{2}\right) \rightarrow  1 ~~(n\rightarrow \infty).\end{eqnarray*} This completes the proof.
\end{proof}

\subsection{$(2-\epsilon)$-approximation for  {Barab\'asi-Albert networks}} \label{sec:ba}
The scale-free property (the power-law tail) has been nicely emulated by the multiple-destination preferential attachment growth model introduced by
Barab\'asi and Albert \cite{ba99}.  Starting with a small number of nodes (who are originally connected with each other), at each time step a new
node enters network $G=(V,E)$, and attaches to $\beta$ existing nodes.
 Each of the existing {nodes} is attached to the new one with a probability that is proportional to its current degree. {Such  a process} is
well-known as the {\it preferential attachment}.
Recall that $|V|=n$. Let $\alpha_{n,k}$ be the fraction of nodes with degree $k$. 
  It is known from \cite{brst01} that for any fixed $\epsilon>0$, and {any $\beta\le k\le n^{1/15}$},\begin{equation}\small{
 \lim_{n\rightarrow
\infty}Pr\left((1-\epsilon)\frac{2\beta(\beta+1)}{k(k+1)(k+2)}\le \alpha_{n,k}\le (1+\epsilon)\frac{2\beta(\beta+1)}{k(k+1)(k+2)}\right)=1.}
\label{t51}
\end{equation}
Note by the construction that each node has a degree of at least $\beta$. Let {$\Gamma$} be the set of all nodes that have a degree of exactly $\beta$. Then
\begin{equation}\Gamma\mbox{~ is an independent set of ~} G,\label{t52}\end{equation} because whenever two nodes are connected, the ``older" one must have a {degree}
 at least $\beta+1$. Note also that for any fixed $\epsilon>0$, {the inequality} $ |E|\le (1+\epsilon/2)n\beta$ holds for big enough $n$. 

\begin{theorem}\label{th:ba} {Consider the PNC model.} {{For any fixed $ \epsilon>0$, with probability tending to one as $n\rightarrow \infty$,
the single pricing with price $\beta$ achieves an approximation ratio of $2-2/(2+\beta)+\epsilon$ for Barab\'asi-Albert network $G$. To be precise,
\begin{equation*}
\lim_{n\rightarrow \infty}Pr\left(\frac{\opt(G)}{n\beta}\le 2-\frac2{(2+\beta)}+\epsilon\right)=1,\end{equation*} where 
$n\beta $ is the revenue obtained by single price $\beta$.}}
\end{theorem}
\begin{proof}
Given an optimal pricing {sequence} $\mathbf p$ for $G=(V,E)$, we construct a charge $c$ on $E$ as follows: At the time a node $u\in V$ {purchases} with price $p$,
{it must have} at least $p$ neighbors, say $v_1,\ldots,v_p$, who have not purchased. We {\em charge} each link $uv_i$ ($1\le i\le p$) with 1. After  the
charge operation is conducted for all nodes,  each link $e\in E$ is charged at most twice (i.e. receives charge at most 2). Define $c(e)=0$ if $e$ is
not charged, $c(e)=1$ if $e$ is charged once, and $c(e)=2$ if $e $ is charged twice. Note that $c(e)=2$ only if the both ends of $e$ purchase at the
same {time (under the same price)}. The charge function $c:E\rightarrow\{0,1,2\}$ satisfies the property that {$c(E)
=\sum_{p\in\mathbf p}p|B(p)|$.} {For
$i=1,2$, let $C_i$ consist of links {$e\in E$} with $c(e)=i$.}

Recall the definition of $\Gamma$ {given above (\ref{t52})}. We denote by $\delta(\Gamma)$ the set of links that are covered by $\Gamma$, and $S$ the set of nodes dominated by $\Gamma$. For each
$u\in S$, let $\delta(u)$ denote the set of links covered by $u$. It is straightforward that
\begin{equation}\delta(\Gamma)\mbox{~is the disjoint union of all~}E_u\equiv\delta(u)\cap \delta(\Gamma), u\in S.\label{t55}\end{equation}
 We also know from (\ref{t51}) and (\ref{t52}) that
 \begin{equation}
 \lim_{n\rightarrow \infty}Pr\left(|\delta(\Gamma)|=\sum_{v\in \Gamma}d(v)\ge (1-\epsilon/2)\frac{2n\beta}{\beta+2}\right)=1\label{t53}
 \end{equation}
 {For any node $u\in S$ with nonempty  {$E_u\cap C_2$}, considering any
 {$uv\in E_u\cap C_2$}, we see that $u$ and $v$ {($\in \Gamma$) purchase under} the same price $p\le d(v)=\beta$. Since $d(u)\ge\beta+|E_u|\ge p+|E_u\cap
C_1|+|E_u\cap C_2|$ and  {$|\delta(u)\setminus C_2|\ge d(u)-p$}, we have  {$|\delta(u)\setminus C_2\setminus(E_u\cap C_1)|=|\delta(u)\setminus
C_2|-|E_u\cap C_1|\ge|E_u\cap C_2|$}. It follows that
\begin{center}
For each $u\in S$, there is a subset $F_u$ of  {$\delta(u)\setminus C_2\setminus(E_u\cap C_1)$} with $|F_u|=|E_u\cap C_2|$ links.
\end{center}
As $F_u$ is disjoint from both $C_2$ and $E_u\cap C_1$, we have $c(e)\le1$ for any $e\in F_u$, and $c(e)=0$ for any $e\in F_u\cap E_u$. This enables
us to modify $c$ to be another charge function $c':E\rightarrow\{0,1,2\}$ such that $c'(E)=c(E) $ and $c'(e)\le1$ for every $e\in\delta(\Gamma)$ as
follows. For each $u\in S$, we increase the charge of each link in $F_u$  by 1, and decrease the charge of each link in $E_u\cap C_2$ by 1. The
resulting charge $c'$ is as desired because, as (\ref{t55}) implies, $\delta(\Gamma)\cap C_2$ is the disjoint union of $E_u\cap C_2$ for all $u\in S$.
Therefore, we obtain  \[\opt(G)=\sum_{p\in\mathbf p}p|B(p)|=c(E)=c'(E)\le2|E|-|\delta(\Gamma)|.\] Using (\ref{t53}), we have with probability tending to 1
(as $n\rightarrow \infty$)\[\opt(G)\le 2(1+\epsilon/2) \beta n-(1-\epsilon/2)\frac{2n\beta}{\beta+2}\le
\left(2-\frac2{\beta+2}+\epsilon\right)n\beta.\]
Observing finally that the single pricing with price $\beta$ obtains revenue $n\beta$ completes the proof.}
\end{proof}

In the special case of $\beta=1$, Barab\'asi-Albert network $G$ is a tree. The approximation ratio $2-2/(\beta+2)=4/3$ stands in contrast to the
ratio 1.5 in Theorem \ref{forest} and Remark \ref{tight1}.

\subsection{Upper and lower bounds for single pricing}\label{upplow}
{Having seen the above constant approximations that single pricing achieves, one may ask: can {best} single pricing always provide good approximations to optimal iterative pricing {for unweighted {networks} with uniform intrinsic values}? The}  following  example shows that, in the worst case,  the best single price can only guarantee at most a fraction $1/({\ln\ln
n})$ of the optimal revenue.
\begin{example}\label{lowerbound}
The network $G$ with $n=k( k!)+1$ nodes   consists of $\sum_{i=1}^ki=k(k+1)/2$ node-disjoint cliques and one special node which is adjacent to all
other nodes, where the number of  $(k!/i)$-cliques is $i$ for each $1\le i\le k$.
\end{example}
In the above instance $G$, there are one node with degree $k(k!)$, which is the special node, and $k!$ nodes with degree $(k!)/i $  for
$i=1,2,\ldots,k$. Recall that $\rv(p)$ denote the revenue under single pricing {$(p)$}. Note that $\rv(k( k!))=k( k!)$, and  $\rv((k!)/i)=(i(k!)+1)\cdot
(k!)/i=(k!)^2+(k!)/i$ for $i=1,\ldots,k$. Then  the best single price is $k!$, which brings a revenue \[\rv(k!)=(k!)^2+k!=\max_{p\ge0}\rv(p).\] On
the other hand the pricing $\mathbf p=(p_1,\ldots,p_{k+1})$ with $p_1=k(k!), p_{i+1}=(k!)/i   $, $i=1,\ldots,k$, brings revenue   $\rv(\mathbf
p)=k(k!)+\sum_{i=1}^{k}(k!)(k!/i-1)=(k!)^2\cdot\sum_{i=1}^{k}(1/i)$. When $k$ tends to infinity,
\begin{equation*}
\frac{\rv(\mathbf p)}{\rv(k!)}=\frac{\sum_{i=1}^k\frac1i}{1+o(1)}\approx1+\ln k=\Theta(\ln\ln n).
\end{equation*}

In complementary to the above example, we show in the following theorem that, with single pricing, one can always assure at least a factor
$1/({1+\ln n})$ of the optimal revenue in unweighed network $G$ {with uniform intrinsic values}. Let $d_1,d_2,\ldots,d_n$ with $d_1\ge d_2\ge\cdots\ge d_n$ be the degree sequence of $G$. 
\begin{theorem}\label{upperbound}
$\opt(G)/\max_{i=1}^n\{id_i\}\le1+\ln n$.
\end{theorem}
\begin{proof}
Since $\sum_{i=1}^{n}d_i\ge \opt(G )$, it suffices to show that \[\max_{i=1,\cdots,n}\{id_i\}\ge\sum_{i=1}^{n}\frac{d_i}{1+\ln n}.\] Suppose on the contrary that
$id_i<\frac{\sum_{j=1}^{n}d_j}{1+\ln n}$  for each $1\le i\le n$. Then we have \begin{eqnarray*}
\sum_{i=1}^{n}d_i<\left(\sum_{i=1}^{n}\frac{1}{i}\right)\cdot \frac{\sum_{i=1}^{n}d_i}{1+\ln n}\Longrightarrow 1+\ln n<\sum_{i=1}^{n}\frac{1}{i},
\end{eqnarray*}
which is a contradiction.
\end{proof}

\section{Conclusion}  \label{sec:conclude}
Our work is an addition to the very limited literature on both pricing with negative network externalities and iterative pricing.
The {model {captures} many {interesting} settings in real-world marketing, and is
usually much more challenging than the positive externality counterpart. The hardness result identifies complexity status of  a fundamental pricing
problem.} The algorithms achieve satisfactory performances in general and several important special settings. {An interesting direction for future research is to narrow the lower and upper bounds on the approximability of the iterative pricing problem with negative externality. Obtaining more accurate estimations for the optimal pricing is a key to reduce the approximation ratios.}

\bibliographystyle{plain}
\bibliography{pricing}

 \bigskip
\appendix
\section*{Appendix: Proof of Theorem \ref{hard}}
{By reduction from the 3SAT problem}, we prove that finding an optimal pricing sequence for the PNC model is NP-hard, even when {the underlying network is an unweighted tree without intrinsic values.}

\section{Construction}

 {Let $I$ be an arbitrary instance of the 3SAT problem, whose input is given by $n$ boolean variables $x_1,x_2,\cdots,x_n$, and $m$ clauses
{$c^j=(x^{j1}\vee x^{j2}\vee x^{j3})$}, $ 1\leq j\leq m$,} where $x^{j\ell}$ is a literal taken from $\{x_1,x_2,\cdots,x_n,\neg x_1,\neg
x_2,$  {$\cdots, \neg x_n\}$}, $1\leq
j\leq m, 1\leq \ell\leq 3$.
To avoid triviality, we assume that
\begin{equation}\label{avoid}
\text{{$m\ge5$}, and for each $i\in[n]$, there exist $j,j'\in[m]$ such that $x_i\in c^j$ and $\neg x_i\in c^{j'}$.}\end{equation}

{From the 3SAT instance  $I$}, we construct an instance   of the {PNC} problem on {tree} $G=(V,E)$ with {unit} weight $w|_{E}=\bfm1$ {and all intrinsic values zero} as follows. Let $R$ (resp. $\bar R$) denote the set of ordered pairs $(i,j)$ such that $i\in[n],j\in[m]$ and $x_i\in c^j$ (resp. $\neg x_i\in c^j$). Clearly, $|R|+|\bar R|=3m$. Let $k=k_0, k_1,k_2,\ldots,k_n$ be integers {satisfying
\begin{equation}\label{e1}
k_i\ge k_{i-1}+6m, \text{ for }i=1,2,\ldots,n,\text{ and }k\ge6m .
\end{equation}
Tree $G=(V,E)$} has    $|V|<3m(k_n^3)+m(k^2+k+1)+(9mk_n^3)$ nodes in total, where $V$ is the disjoint union of the node sets of $3m$ variable
gadgets, $m$ clause gadgets and one connection gadget.
\begin{itemize}
\item For every $(i,j)\in R$, i.e., $x_i\in c^j$, there is a {\em variable gadget} $X_i^j=(V_i^j,E_i^j)$ which is a tree rooted at node $x_i^j$  (see   Figure \ref{Figure.1}(a)). Node set $V_i^j$ with $|V_i^j|=k_i^3-2k_i^2+1<k_n^3$ is the disjoint union of {four sets} $\{x_i^j\}$, $V_{i1}^j$, $V_{i2}^j$ and {$V_{i3}^j$}, where $V_{ih}^{j}$, $h=1,2,3$, consists of nodes  in $X_i^j$ at distance $h$ from $x_i^j$.
 \begin{itemize}
 \item {\em Literal node} $x_i^j$, which simulates literal $x_i$, has degree $ k_i-2$ in $X_i^j$.
 \item $V_{i1}^j$ consists of the $k_i-2$ neighbors of $x_i^j$ in $X_i^j$, all having degree $k_i$.
 \item $V_{i2}^j$ consists of $(k_i-2)(k_i-1)$ nodes, all having degree $k_i+1$.
\item $V_{i3}^j$ consists of the $(k_i-2)(k_i-1)k_i$  leaves of  $X_i^j$.
\end{itemize}

\item For every $(i,j)\in\bar R$, i.e., $\neg x_i\in c^j$, there is a {\em variable gadget} $\bar X_i^j=(\bar V_i^j,\bar E_i^j)$ which is a tree rooted at node $\bar x_i^j$  (see   Figure \ref{Figure.1}(b)). Node set $\bar V_i^j$ with  $|\bar V_i^j|=k_i^2<k_n^3$ is the disjoint union of {three sets} $\{\bar x_i^j\}$, $\bar V_{i1}^j$ and  $\bar V_{i2}^j$, where $\bar V_{ih}^j$, $h=1,2$, consists of nodes in $\bar X_i^j$ at distance $h$ from $\bar x_i^j$.
\begin{itemize}
\item {\em Literal node} $\bar x_i^j$, which simulates literal $\neg x_i$, has degree $ k_i-1$ in $\bar X_i^j$.
\item $\bar V_{i1}^j$ consists of the $k_i-1$ neighbors of $\bar x_i^j$ in $\bar X_i^j$, all having degree $k_i+1$.
\item $\bar V_{i2}^j$ consists of the $(k_i-1)k_i$   leaves of $\bar X_i^j$.
\end{itemize}

\item For each clause $c^j$, there is a {\em clause gadget} $C^j=(V^j,E^j)$ which is a tree rooted at {node} $c^j$ (see   Figure~\ref{Figure.1}(c)). Node set $V^j $ with $|V^j|=k^2+k+1$ is the disjoint union of {three sets} $\{c^j\}$, $V^{j1}$ and {$V^{j2}$}, where $V^{jh}$, $h=1,2$, consists of nodes in $C^j$ at distance $h$ from $c^j$.
 \begin{itemize}
 \item {\em Clause node} $c^j $, which simulates the clause, has degree $k$ in $C^j$.
  \item $V^{j1}$ consists of the $k$ neighbors of $c^j$ in $C^j$, all having degree $k+1$.
\item $V^{j2}$ consists of the $k^2$   leaves of $  C^j$.
\end{itemize}

\item For any literal $x_i$ and clause $c^j$ with $x_i\in c^j$, there is a link joining literal node $x_i^j$ and clause node $c^j$. For any literal $\neg x_i$ and clause $c^j$ with $\neg x_i\in c^j$, there is a link joining literal node $\bar x_i^j$ and clause node $c^j$.
\item  There is a {\em connection gadget} $S=(V_S,E_S)$ which is a star centered at node~$s$. The {\em connection node} $s$ has degree $ 9mk_n^3-1$ in $S$, and is adjacent to every clause node of~$G$.
\end{itemize}

\begin{figure}[h!]
\begin{center}
\includegraphics[width=16cm]{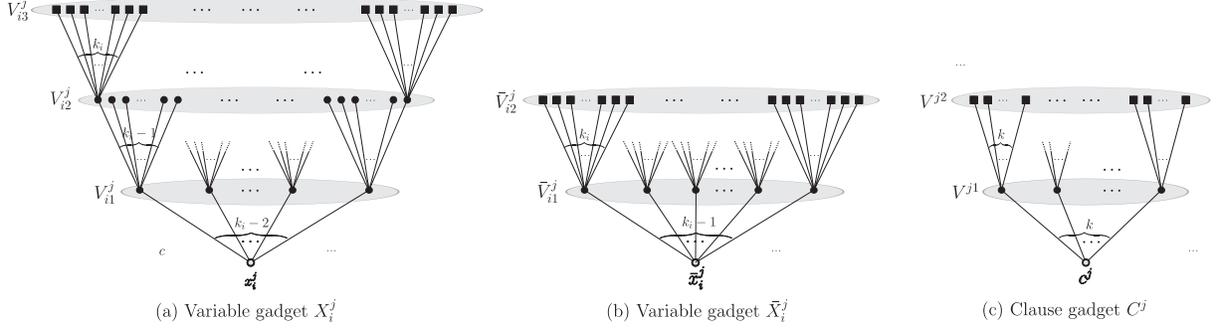}
\caption{
Literal nodes and clause nodes are represented by {circles}, while auxiliary nodes are represented by {solid disks} and (pendant) squares. }
\label{Figure.1}
\end{center}\end{figure}

Obviously, the above construction of $G=(V,E)$ can be done in polynomial time. It is easy to check that $G$ is a tree. In particular, all the $4m+1$
node-disjoint  gadgets {(recall that $|R|+|\bar R|=3m$)} are connected by $4m$ links adjacent to clause nodes, where each clause node $c^j$ has exactly four neighbors outside $C^j$,
three being literal nodes and one being the connection node $s$.  Let $E_{cl}$ denote the {set of} $3m$ links connecting clause nodes and literal nodes, and
$E_{cs}$ denote the {set of} $m$ links connecting clause nodes and {connection node} $s$. Then
\begin{align}
&\text{$V$ is the disjoint union of $ \bigcup_{(i,j)\in R}V_i^j,\bigcup_{(i,j)\in \bar R}\bar V_i^j,\bigcup_{j\in[m]}V^j$ and $V_S$.}\label{node}\\
&\text{$E$ is the disjoint union of $ \bigcup_{(i,j)\in R}E_i^j,\bigcup_{(i,j)\in \bar R}\bar E_i^j,\bigcup_{j\in[m]}E^j$, $E_S$, $E_{cl}$, and
$E_{cs}$.}\label{edge}
\end{align}

For any {node $v\in V$}, let {$d(v)$}   denote the degree of $v$ in $G$. The connection node $s$ has degree
\begin{equation}\label{es}
d(s)=|E_S|+|E_{cs}|=9mk_n^3-1+m>9mk_n^3.\end{equation}

Note that each literal node   has exactly one neighbor outside the variable gadget containing it, which is a
clause node. Therefore, for any {$(i,j)\in R$}, we have
{\begin{align}
d(x_i^j)=k_i-1, \;\; d(v)=k_i\text{ for all }v\in V_{i1}^j,\; \text{ and }   d(v)=k_i+1\text{ for all }v\in V_{i2}^j.\label{ev1}\end{align}}

For  any  {$(i,j)\in \bar{R}$}, we have
\begin{align} d(\bar x_i^j)=k_i, \text{ and }d(v)=k_i+1\text{ for all }v\in \bar V_{i1}^j.\label{ev2}
\end{align}

For every $j\in[m]$, we have

\begin{equation}\label{ec}
d(c^j)=k+4,\text{ and }d(v)=k+1\text{  for all }v\in V^{j1}.\end{equation}

Other nodes, i.e., those {not} mentioned in (\ref{es}) -- (\ref{ec}), are exactly
leaves of $G$. It is worthwhile noting from (\ref{e1}) that  all the non-leaf consumers in the variable gadgets have much larger degrees than the
non-leaf consumers in the clause gadgets. This   permits us to consider the former consumers before the latter ones.

\section{Pricing}
Given {any integer}   pricing sequence $\mathbf{p}=(p_1,p_2,\ldots, p_\tau)$ for $G$, let $\xi=\min\{t: p_t\in \mathbf{p}\text{ and }p_t\leq k+3\}$ be the first time
that the price is equal to or lower than  $k+3$. Note that if $p_1=d(s)$, then {at time 1, only $s$ purchases, and} before time $\xi$, no consumer in the clause gadgets has purchased,
i.e. \begin{equation}\label{early}B_{\xi-1}(\mathbf p)\cap (\cup_{j=1}^m V^j)=\emptyset.\end{equation}


\begin{clm}\label{c4}
 {Let $D^j=\{x_i^j:x_i\in c^j,i\in[n]\}\cup\{\bar x_i^j:\neg x_i\in
c^j,i\in[n]\}$}. If   $p_1=d(s)$, then the following holds for each {$j\in [m]$}.
\begin{mylabel}
\item   {E}ither $\rv_{C^j}(\mathbf p)\le  |E^j|+3${,}  or {$|E^j|+k+1\le\rv_{C^j}(\mathbf p)\le|E^j|+k+3$.}
\item   {If $\rv_{C^j}(\mathbf p)\ge |E^j|+k+1$  then
$D^j \setminus B_{\xi-1}(\mathbf p) \neq \emptyset$ and $\mathbf p\cap[k+1,k+3]\ne\emptyset$}. 
\item {If $D^j \setminus B_{\xi-1}(\mathbf p) \neq \emptyset$ and $ \mathbf p\cap[k,k+3]=\{k+1\}$, then $\rv_{C^j}(\mathbf p)= |E^j|+k+1$.}
\end{mylabel}
\end{clm}
\begin{proof}
{Recall that $\nu_t(v,\mathbf p)$  denotes the (total) value of the product at time $t$ for consumer $v\in V$. As $p_1= d(s)=\nu_1(s,\mathbf p)>d(v)=\nu_1(v,\mathbf p)$ for all $v\in V\setminus\{s\}$, only consumer $s$ purchases at time 1. Then for any $t\in[2,\tau]$  it holds that $\nu_t(c^j,\mathbf p)\le\nu_2(c^j,\mathbf p)=k+3$.
 Thus $c^j$ can only purchase at a price no greater than $k+3$.}

{In case of $c^j$ purchasing at some price in $[k+2,k+3]$, either  all consumers in $ V^{j1}$ purchase at some price within $[2,k]$, in which case $\rv_{C^j\setminus\{c^j\}}(\mathbf p)\le k^2=|E^j|-k$; or   all consumers in $C^j\setminus\{c^j\}$ purchase at price 1, in which case $\rv_{C^j\setminus\{c^j\}}(\mathbf p)=k+k^2=|E^j|$ . So either $\rv_{C^j}(\mathbf p)\le  |E^j|+3$ or $\rv_{C^j}(\mathbf p)\in\{|E^j|+k+2,|E^j|+k+3\}$.}

{In case of $c^j$ purchasing at   price  $k+1$, all consumers in $V^{j1}$ purchase at price $k+1$, giving  $\rv_{C^j}(\mathbf p)=(k+1)(1+k)=k^2+2k+1=|E^j|+k+1$.}

Consider now the case of $c^j$ not purchasing at any price above $k$. Note first it is possible that $c^j$ never purchases at all under $\mathbf p$, and {that} consumers in $V^{j1}$ will never purchase at a price higher than $k+1$. If one and thus all consumers in $V^{j1}$ purchase before $c^j$, then $\rv_{C^j}(\mathbf p)$ is maximized when all consumers in $V^{j1}$ purchase at price $k+1$, saying  $\rv_{C^j}(\mathbf p)\le (k+1)k+3=|E^j|+3$. If none of {the} consumers in $V^{j1}$ {purchases} before $c^j$, then $\rv_{C^j}(\mathbf p)\le \max\{k(1+k),1+k+k^2\}=|E^j|+1$.

 {Hence we see that (i) holds, and $\rv_{C^j}(\mathbf p)\ge |E^j|+k+1$  only if $c^j$ purchases under some price  $p_t\in[k+1,k+3]$ at time $t$. Recalling the time point $\xi$ defined at the beginning of this section, we have $t\ge\xi$, which implies $D^j \setminus B_{\xi-1}(\mathbf p) \neq \emptyset$. So (ii) is  valid.}

 {Suppose now that $D^j \setminus B_{\xi-1}(\mathbf p) \neq \emptyset$. Recall from (\ref{early}) that $c^j$ does not purchase before time $\xi$.  It follows that $\nu_{\xi}(c^j,\mathbf p)\ge k+1$. If $ \mathbf p\cap[k,k+3]=\{k+1\}$, then $p_{\xi}=k+1$. It follows that $c^j$ and all consumers in $V^{j1}$ purchase at price $p_{\xi}=k+1$, yielding (iii).}\end{proof}

 {Note from (\ref{ev1}) and (\ref{ev2}) that for} any $i\in[n]$, both $\max\{d(v): (i,j)\in R,\,v\in V_i^j\} $ and {$\max\{d(v): (i,j)\in \bar R$, $v\in \bar V_i^j\} $ }are upper bounded by $k_i+1$;
thus none of consumers in $ \cup_{j:(i,j)\in R}V_i^j$ and $ \cup_{j:(i,j)\in \bar R}\bar V_i^j$ purchases when the price is above $k_i+1$.
Furthermore, the following two claims can be easily checked {by charging (a part of) revenue obtained at a vertex to a subset of edges incident with it, where each edge receives a charge of 1.}
\begin{clm}\label{c1}For any $(i,j)\in R$, the following hold:
 \begin{mylabel}
 \item  if $\mathbf{p}\cap [k_i-1,k_i+1 ]=\{k_i+1,k_i-1\}$,   then $\rv_{\xi-1,V_i^j}(\mathbf p)=|E_i^j| +1$ and $x_i^j\in B_{\xi-1}(\mathbf p)$;
  \item  if $\mathbf{p}\cap [k_i-1,k_i+1 ]=\{k_i \}$,   then $\rv_{\xi-1,V_i^j}(\mathbf p)=|E_i^j| $ and $x_i^j\not\in B_{\xi-1}(\mathbf p)$.
  \item {if $\mathbf{p}\cap [k_i-1,k_i+1]\in \{\{k_i+1,k_i-1\},\{k_i \}\}$, then $\rv_{V_i^j}(\mathbf p)\le|E_i^j|+1$.}
     \end{mylabel}
     \end{clm}
\begin{proof}{To see (i), we consider the time, say $t$, when price $p_t=k_i+1$ is announced, all consumers in $V^j_{i2}$ purchase, and others in $V^j_i\setminus V^j_{i2}$ do not. We charge the revenue $k_i+1$ obtained at each consumer of $V^j_{i2}$ to the $k_i+1$ edges incident with it. Next, at time $t+1$, price $p_{t+1}=k_{i}-1$ is announced, and only $x^j_i$ purchases, because the product value is $k_i-1$ for $x^j_i$, and 1 (resp. 0) for each consumer in $V^j_{i1}$ (resp. $V^j_{i3}$) at that time. Now we charge the $k_i-2$ edges in $E^j$ that are incident with $x^j_i$. So all edges in {$E^j_i$} are charged and 1 revenue is left {(this amount corresponds to the edge that connects $x_i$ and the clause gadget $C^j$)}, which gives (i), as after $x^j_i$'s purchase the product values 0 for all consumers in $V^j_{i1}\cup V^j_{i3}$.}

{To see (ii), {note first that only consumers in $V^j_{i1}\cup V^j_{i2}$ purchase under price $k_i$}. For each $v\in V^j_{i1}$, we charge the $k_i$ edges incident with $v$; for each $v\in V^j_{i2}$, we charge the $k_i$ pendant edges incident with $v$. All edges of {$E^j_i$} have been charged and no revenue is left. Now, the product values 1 for $x^j_i$ and 0 for all consumers in $V^j_{i3}$. Hence (ii) holds.}

{Statement (iii) is straightforward from the proofs of (i) and (ii).}
\end{proof}

\begin{clm}\label{cneg}For any  {$(i,j)\in \bar{R}$}, the following hold:
 \begin{mylabel}
 \item  if $\mathbf{p}\cap [k_i-1,k_i+1 ]=\{k_i+1,k_i-1\}$,   then $\rv_{\xi-1,\bar{V}_i^j}(\mathbf p)=|\bar{E}_i^j| $ and  $\bar x_i^j\not\in B_{\xi-1}(\mathbf p)$;
  \item  if $\mathbf{p}\cap [k_i-1,k_i+1 ]=\{k_i \}$,   then $\rv_{\xi-1,\bar{V}_i^j}(\mathbf p)=|\bar{E}_i^j|+1$ and $\bar x_i^j \in B_{\xi-1}(\mathbf p)$.
   \item {if $\mathbf{p}\cap [k_i-1,k_i+1]\in \{\{k_i+1,k_i-1\},\{k_i \}\}$, then $\rv_{\bar V_i^j}(\mathbf p)\le|\bar E_i^j|+1$.}
     \end{mylabel}
     \end{clm}
\begin{proof} {In proving (i), for each consumer in $\bar V^j_{i1}$, we charge the $k_i+1$ edges incident with it. In proving (ii), for  each consumer in $\bar V^j_{i1}$, we charge the $k_i $ pendant edges incident with it; for $\bar x^j_i$, we charge the $k_i-1$ edges in {$\bar E_i^j$} that are incident with it. Statement (iii) is then instant.}
\end{proof}

{In the rest of this section we discuss the properties of normal pricing sequences.}

\begin{clm} \label{c3}
For any $i\in[n]$, if {$\mathbf p$ is normal and} $\mathbf{p}\cap [k_i-1,k_i+1]\notin \{\{k_i+1,k_i-1\},\{k_i \}\}$, then one of the following holds:
\begin{mylabel}
\item  $\mathbf{p}\cap [k_i,k_i+1 ]=\emptyset$, in which case {$\rv_{V_i^j}(\mathbf p)< |E_i^j|-12m+1 $ for every every $j\in[m]$ with $(i,j)\in R$, and  $\rv_{\bar V_i^j}(\mathbf p)< |\bar E_i^j| -12m+1$ for every every $j\in[m]$ with $(i,j)\in \bar R$.}
 \item  $\mathbf{p}\cap [k_i-1,k_i+1]=\{k_i+1\}$ and  $\mathbf{p}\cap [2,k_i-2]\ne\emptyset$,    {in which case} {$\rv_{V_i^j}(\mathbf p)\le |E_i^j|-6m+5 $} for every $j\in [m]$ with $(i,j)\in R$, and  $\rv_{\bar V_i^j}(\mathbf p)\le|\bar E_i^j|$ for every $j\in [m]$ with $(i,j)\in\bar R$.
 \item  $\mathbf{p}\cap [k_i-1,k_i+1]=\{k_i+1\}$ and
 $\mathbf{p}\cap [1,k_i-2]\subseteq\{1\}$,  {in which case} $\rv_{V_i^j}(\mathbf p)\le |E_i^j|+1 $ for every every $j\in[m]$ with $(i,j)\in R$, and  {$\rv_{\bar V_i^j}(\mathbf p)\le |\bar E_i^j| +1$} for {every} $j\in[m]$ with $(i,j)\in \bar R$. \end{mylabel}
\end{clm}
\begin{proof}
If $\mathbf{p}\cap [k_i,k_i+1 ]=\emptyset$, then no consumer in $(\cup_{j:(i,j)\in R}V_i^j)\cup(\cup_{j:(i,j)\in \bar R}\bar V_i^j)$ purchases at a price
higher than $k_i-1$. 
It follows that {$\rv_{V_i^j}(\mathbf p)\le(k_i-1)|V^j_{i2}\cup V^j_{i1}\cup \{x_i^j\}|$}=$1+ |E_i^j|-|V^j_{i2}|=|E_i^j|-(k_i-2)(k_i-1)+1$ for every   $j\in[m]$ with $(i,j)\in R$ 
and  {$\rv_{\bar V_i^j}(\mathbf p) \le(k_i-1)|\bar V^j_{i1}\cup \{\bar x_i^j\}|=|\bar E_i^j|-|\bar
V_{i1}^j|=|\bar E_i^j|-(k_i-1)$} for every   $j\in[m]$ with $(i,j)\in \bar R$. 
 Now $(k_i-2)(k_i-1)-1> k_i-1\ge12m-1$, {which is implied by (\ref{e1}),}   gives (i).

   {It remains to consider the case where there exist $p_t$ ($1\le t\le\tau$) that is the maximum price in $\mathbf{p}\cap[k_i,k_i+1]\ne\emptyset$. If $p_t=k_i$, then the maximality of $p_t$ together with (\ref{e1}) and (\ref{ev1}) -- (\ref{ec}) implies that all consumers in $V^j_{i1}\cup V^j_{i2}$ with $(i,j)\in R$ and those in $\{\bar x^j_i\}\cup\bar{V}^j_{i1}$ with $(i,j)\in\bar R$ would purchase under price $k_i$ at time $t$. Note from (\ref{es}) that consumer $s$ must have purchased by time $t$. After time $t$ any consumer without the product has value at most $k_{i-1}+3<k_i-1$ (recall (\ref{ev1}) -- (\ref{ec}) and (\ref{e1})). It follows from normality of $\mathbf{p}$ that $k_i-1\not\in\mathbf p$, enforcing $\mathbf{p}\cap [k_i-1,k_i+1]=\{k_i\}$, a contradiction to the condition $\mathbf{p}\cap [k_i-1,k_i+1]\notin \{\{k_i+1,k_i-1\},\{k_i \}\}$ of the claim. Thus  $p_t=k_i+1$, and all consumers in $  V^j_{i2}$ with $(i,j)\in R$ and those in $ \bar{V}^j_{i1}$ with $(i,j)\in\bar R$   purchase under price $k_i+1$ at time $t$, bringing about revenues $|E_i^j|-(k_i-2)$ and $|\bar E_i^j|$, respectively. Notice again that $s$ has purchased by time $t$. After time~$t$, any consumer without the product has value at most $k_i-1$,  which along with the normality of $\mathbf p$ gives $k_i\not\in\mathbf p$}. In turn $\mathbf{p}\cap [k_i-1,k_i+1]\ne\{k_i+1,k_i-1\}$ implies $\mathbf{p}\cap [k_i-1,k_i+1]=\{k_i+1\}$.

   As $k_i-1\not\in\mathbf p$, the normality of $\mathbf p$ enforces $\mathbf p\cap[k_{i-1}+4,k_i-1]=\emptyset$. It follows from (\ref{e1})   {that
   \[\rv_{V_i^j}(\mathbf p)\le |E_i^j|-(k_i-2)+(k_{i-1}+3)\le |E_i^j|-6m+5.\]}

  In case of $\mathbf{p}\cap [2,k_i-2]\ne\emptyset$, before variable nodes $\bar x^j_i$ with $(i,j)\in\bar R$ purchase {(possibly) at price $1$ or 0,} all clause nodes have purchased under some price in $\mathbf{p}\cap [2,k_i-2]$. It follows that all these $\bar x^j_i$ with $(i,j)\in\bar R$  can only purchase at price 0, yielding (ii).

 In case of $\mathbf{p}\cap [2,k_i-2]=\emptyset$, we have $\mathbf{p}\cap [1,k_i-2]\subseteq\{1\}$, implying $\rv_{\bar V_i^j}(\mathbf p)\le |\bar E_i^j| +1$ and hence (iii).

  Due to the above analysis, it can also be observed that the three situations {stated}  in this claim  are all the possible ones. \end{proof}

 Combining Claims \ref{c1}(iii), \ref{cneg}(iii) and \ref{c3} we obtain the following corollary.
 \begin{clm}\label{comb}
{If $\mathbf p$ is normal, then  $\rv_{V_i^j}(\mathbf p)\le|E_i^j|+1$ for all $(i,j)\in R$ and $\rv_{\bar V_i^j}(\mathbf p)\le|\bar E_i^j|+1$ for all $(i,j)\in\bar R$.}\end{clm}
 \begin{clm} If {$\mathbf p$ is normal and} $\rv(\mathbf p)>|V|$, then $p_1=d(s)$.\label{ds}\end{clm}
\begin{proof}Suppose {to the contrary} that $p_1\ne d(s)$. By normality of $\mathbf p$, we have $p_1<d(s)$, and {furthermore} $p_1\le k_n+1$ {(recalling (\ref{e1}) and (\ref{es})--(\ref{ec})).} It follows that either $p_1=1$, giving $\rv(\mathbf p)=|V|$, or $p_1\ge2$, giving {
\begin{eqnarray*}\rv(\mathbf p)=\rv_{\{s\}}(\mathbf p)+\rv_{V\setminus\{s\}}(\mathbf p)\le (k_n+1)+2(|E|-|E_S|)=
 (k_n+1)+2(|V|\!-\!1\!-\!|E_S|)< 9mk_n^3<d(s)<|V|,\end{eqnarray*} where the third last inequality uses the fact that $|V|<3m(k_n^3)+m(k^2+k+1)+(9mk_n^3)$ and $|E_S|=9mk_n^3-1$}.\end{proof}


\section{Final proof}
{Having finished all necessary preparations}, we are ready to establish the close relation between 3SAT instance $I$ and the PNC instance on
tree $G$.

\bigskip\noindent{\bf Theorem \ref{hard}.}  (Restated)\quad{\em {{In} the PNC model, computing the  optimal pricing sequence is NP-hard, even when the underlying network is an unweighted tree and all the intrinsic values are zero.}}

\begin{proof}  Let $\opt(G)$ denote the optimal objective value of the {PNC instance on tree $G=(V,E)$}. Define
\begin{equation*}L=|E|+(k-2)m
\end{equation*}
  To establish the NP-hardness of the pricing problem, it suffices to prove that
{\em $\opt(G)\geq L$ if and only if {the 3SAT instance} $ I$  is satisfiable.}

 \bigskip\noindent{\em The ``if'' part.}\quad Suppose that $I$ has a satisfactory truth assignment $\pi$ with $\lambda$ variables assigned ``TRUE" and the remaining $n-\lambda$ variables  assigned
``FALSE". Let $\mathbf{p}=(p_1,p_2,\ldots,p_{2n-\lambda+1},p_{2n-\lambda+2})$ be a solution to {the PNC instance on $G$}  such that
\begin{itemize}
\item $p_1=d(s)$;
\vspace{-2mm}\item There are one or two prices for each variable gadget depending on whether the variable is assigned ``TRUE'' or ``FALSE'' in $\pi$:
 {if $x_{i}$ is assigned ``TRUE" then $ k_i \in \mathbf{p}$},
 if $x_{i}$ is
assigned ``FALSE", then  $\{k_i+1,k_i-1\}\subset\mathbf{p}$;
\vspace{-2mm}\item There is a common price for the $m$ clause gadgets:
  $p_{2n-\lambda+2}=k+1\in \mathbf p$.
\end{itemize}
According to Claims \ref{c1} and  \ref{cneg}, we have
\[\sum_{(i,j)\in R}\rv_{V_i^j}(\mathbf p)+\sum_{(i,j)\in \bar R}\rv_{\bar V_i^j}(\mathbf p)\ge\sum_{(i,j)\in
R}|E_i^j|+\sum_{(i,j)\in \bar R}|\bar E_i^j|.\] Furthermore, the satisfiability implies that $D^j \setminus B_{\xi-1}(\mathbf p) \neq \emptyset$. {Therefore, the condition in Claim  \ref{c4}(iii) holds for every
$j\in [m]$, giving $\rv_{C^j}(\mathbf p)=|E^j|+k+1$} for every $j\in[m]$. It follows from (\ref{node}) that the pricing sequence $\mathbf p$ assures a revenue
\begin{eqnarray*}
\rv(\mathbf p)&=&\rv_{V_S}(\mathbf p)+\sum_{(i,j)\in R}\rv_{V_i^j}(\mathbf p)+\sum_{(i,j)\in \bar R}\rv_{\bar V_i^j}(\mathbf p)+\sum_{j\in[m]}\rv_{C^j}(\mathbf p)\\
&\ge&d(s)+\sum_{(i,j)\in R}|E_i^j|+\sum_{(i,j)\in \bar R}|\bar E_i^j|+\sum_{j\in[m]}(|E^j|+1+k)\\
&=&|E_S|+|E_{cs}|+\sum_{(i,j)\in R}|E_i^j|+\sum_{(i,j)\in \bar R}|\bar E_i^j|+\sum_{j\in[m]}(|E^j|+1+k) \end{eqnarray*}
Now from (\ref{edge}) we derive $\rv(\mathbf p)\ge |E|-|E_{cl}|+m(1+k)=|E|-3m+m(k+1)= L$, proving the {``if''} part.

\bigskip\noindent{\em The ``only if'' part.}\quad  Suppose now $\opt(G)\geq L $. Due to Observation \ref{o3}, there exists a normal pricing sequence $\mathbf{p}=(p_1,p_2,\ldots,p_\tau)$
whose objective {value $\rv(\mathbf p)$} is at least $L$.   As $k\ge5m$, which implies {$L>|E|+1=|V|$},  we derive from Claim \ref{ds} that $p_1=d(s)$, which validates the
subsequent application of Claim \ref{c4}.

 If $\rv_{C^{j_0}}(\mathbf p)<|E^{j_0}|+k+1$ for some $j_0\in [m]$, it can be deduced from Claim \ref{c4}(i) that  $\rv_{C^{j_0}}(\mathbf p)\le|E^{j_0}|+3${.}
Recalling
  (\ref{node}), we derive {from Claims \ref{comb} and \ref{c4}(i) that
   \begin{eqnarray*}
\rv(\mathbf p)&=&\rv_{V_S}(\mathbf p)+\sum_{(i,j)\in R}\rv_{V_i^j}(\mathbf p)+\sum_{(i,j)\in \bar R}\rv_{\bar V_i^j}(\mathbf p)+\sum_{j\in[m]}\rv_{C^j}(\mathbf p)\\
&\le&|E_S|+|E_{cs}|+\sum_{(i,j)\in R}(|E_i^j|+1)+\sum_{(i,j)\in \bar R}(|\bar E_i^j|+1)+\sum_{j\in[m]\setminus\{j_0\}}(|E^j|+k+3)+|E^{j_0}|+3.
\end{eqnarray*}
 Recalling (\ref{edge}), we have
 \begin{eqnarray*}
\rv(\mathbf p)&\le&|E| -|E_{cl}|+(|R|+|\bar R|) +(m-1)(k+3)+3\\
&=&|E|-3m+3m +(k-2)m+5m-k\\
&=&L-(5m-k).
\end{eqnarray*}
Then $k>5m$} implies $\rv(\mathbf p)<L$, a contradiction. Thus for every $j\in[m]$ we have $\rv_{C^j}(\mathbf p)\ge|E^j|+k+1$, {which along with Claim \ref{c4}(ii) implies $D^j \setminus B_{\xi-1}(\mathbf p) \neq \emptyset$ and  $\mathbf{p}\cap
[k+1,k+3]\ne\emptyset$.}

Suppose that there exists $i_0\in[n]$ such that  $\mathbf{p}\cap [k_{i_0}-1,k_{i_0}+1]\notin \{\{k_{i_0}+1,k_{i_0}-1\},\{k_{i_0} \}\}$. Recall {from (\ref{avoid})} that
there exists $j_0\in[m]$ such that $(i_0,j_0)\in R$. Notice from $\mathbf{p}\cap
[k+1,k+3]\ne\emptyset$ that $\mathbf{p}\cap
[2,k_{i_0}-2]\ne\emptyset$, {because $[k+1,k+3]\subseteq [2,k_{i_0}-2]$ as guaranteed by (\ref{e1})}.  If $\mathbf{p}\cap [k_{i_0}-1,k_{i_0}+1]=\{k_{i_0}+1\}$,   {then Claim \ref{c3}{(ii)} implies  that $\rv_{V^{j_0}_{i_0}}(\mathbf p)\le|E^{j_0}_{i_0}|-6m+5$ and further that
 \begin{eqnarray*}
 \rv(\mathbf p)&=&\rv_{V_S}(\mathbf p)+\sum_{(i,j)\in R}\rv_{V_i^j}(\mathbf p)+\sum_{(i,j)\in \bar R}\rv_{\bar V_i^j}(\mathbf p)+\sum_{j\in[m]}\rv_{C^j}(\mathbf p)\\
&\le&|E_S|+|E_{cs}|+\sum_{(i,j)\in R\setminus\{(i_0,j_0)\}}(|E_i^j|+1)+(|E_{i_0}^{j_0}|-6m+5)+\sum_{(i,j)\in \bar R}(|\bar E_i^j|+1)+\sum_{j\in[m] }(|E^j|+k+3) \\
&=&|E| -|E_{cl}|+(|R|+|\bar R|-1)-6m+5 +m(k+3) \\
&=&|E|-3m+(3m-1)  +m(k-2)+5-m\\
&=&L+4-m.
\end{eqnarray*}
Then $m\ge5$ implies a contradiction to $\rv(\mathbf p)\ge L$, which  reduces us to the case $\mathbf p\cap[k_{i_0},k_{i_0}+1]=\emptyset$ and  $\rv_{V^{j_0}_{i_0}}(\mathbf p)<|E^{j_0}_{i_0}|-12m+1$ as stated in  Claim \ref{c3}(i). {Since $|E^{j_0}_{i_0}|-12m+1$ is obviously smaller than $|E^{j_0}_{i_0}|-6m+5$, we still have $\rv(\mathbf p)<L$.}
 The contradiction shows that  no such an $i_0\in[n]$ exists, and therefore} the conditions in Claims \ref{c1} and \ref{cneg} {hold}. This enables us to
  construct a truth assignment $\pi$ as follows: for each
$1\leq i\leq n$, if $\mathbf{p}\cap [k_i-1,k_i+1]=\{k_i+1,k_i-1\}$, we assign ``FALSE" to variable $x_i$. Otherwise, that is $\mathbf{p}\cap
[k_i-1,k_i+1]=\{k_i\}$, we assign $x_i$  ``TRUE". As {argued} above, $D^j \setminus B_{\xi-1}(\mathbf p) \neq \emptyset$ for all $j\in[ m]$.
Therefore $\pi$ is indeed a satisfactory truth assignment for $I$. This completes the {``only if"} part and the whole proof of Theorem
\ref{hard}. \end{proof}

\end{document}